\documentclass[11pt]{article}

\usepackage{anysize,amsmath,amsthm,amsfonts,amssymb,graphics,graphicx,enumerate,multirow}

\usepackage{lscape}

\usepackage{tikz}
\usetikzlibrary{arrows,shapes}
\usepackage{subfigure}
\usetikzlibrary{patterns}
\usetikzlibrary{decorations.pathreplacing}

\newtheorem{definition}{Definition}
\newtheorem{theorem}{Theorem}
\newtheorem{lemma}{Lemma}
\newtheorem{proposition}{Proposition}
\newtheorem{claim}{Claim}
\newtheorem{remark}{Remark}

\newtheorem{assumption}{Assumption}
\newtheorem*{assumption*}{Assumption}

\newcommand{\blankline}{\vspace{0.7cm}}
\newcommand{\tinyblank}{\vspace{0.3cm}}

\newcommand{\ml}{\mathcal{L}}
\newcommand{\me}{\mathcal{E}}

\newcommand{\eps}{{\epsilon}}
\newcommand{\cR}{{\mathcal{R}}}
\newcommand{\tcR}{{\tilde{\mathcal{R}}}}
\newcommand{\hcR}{{\widehat{\mathcal{R}}}}
\newcommand{\cK}{{\mathcal{K}}}
\newcommand{\cV}{{\mathcal{V}}}
\newcommand{\cP}{{\mathcal{P}}}
\newcommand{\cU}{{\mathcal{U}}}
\newcommand{\tcV}{{\tilde{\mathcal{V}}}}
\newcommand{\tV}{{\tilde{V}}}
\newcommand{\Ev}{{\mathcal{B}}}
\newcommand{\prob}{{\Pr}}
\newcommand{\ind}{\mathbb{I}}
\newcommand{\cN}{{\mathcal{N}}}

\newcommand{\EvD}{{\mathcal{D}}}
\newcommand{\EvF}{{\mathcal{F}}}

\newcommand{\gammabf}{\mathbf{\gamma}}
\newcommand{\alphabf}{\mathbf{\alpha}}
\newcommand{\argmax}{\textrm{argmax}}

\newcommand{\core}{\mathcal{C}}

\newcommand{\typelab}{\mathcal{T_{\ml}}}
\newcommand{\typeemp}{\mathcal{T_{\me}}}
\newcommand{\E}{\textrm{E}}

\newcommand{\Mid}{\left\vert\vphantom{\frac{1}{1}}\right.}
\newcommand{\real}{{\mathbb{R}}}
\def\D{\mathrm{d}}
\newcommand{\tn}{{\tilde{n}}}
\newcommand{\utheta}{{\underline{\theta}}}
\newcommand{\btheta}{{\overline{\theta}}}
\newcommand{\bk}{{\bar{k}}}
\newcommand{\uk}{{\underline{k}}}
\newcommand{\reals}{{\mathbb{R}}}

\makeatletter
\newsavebox\myboxA
\newsavebox\myboxB
\newlength\mylenA

\newcommand*\xoverline[2][0.82]{%
    \sbox{\myboxA}{$\m@th#2$}%
    \setbox\myboxB\null
    \ht\myboxB=\ht\myboxA%
    \dp\myboxB=\dp\myboxA%
    \wd\myboxB=#1\wd\myboxA
    \sbox\myboxB{$\m@th\overline{\copy\myboxB}$}
    \setlength\mylenA{\the\wd\myboxA}
    \addtolength\mylenA{-\the\wd\myboxB}%
    \ifdim\wd\myboxB<\wd\myboxA%
       \rlap{\hskip 0.5\mylenA\usebox\myboxB}{\usebox\myboxA}%
    \else
        \hskip -0.5\mylenA\rlap{\usebox\myboxA}{\hskip 0.5\mylenA\usebox\myboxB}%
    \fi}
\makeatother

\newtheorem*{theorem*}{Theorem}

\begin{document}

\begin{titlepage}
\title{The size of the core in assignment markets}
\author{Yash Kanoria \and Daniela Saban \and Jay Sethuraman\thanks{
All the authors are at Columbia University. Emails: \texttt{\{ykanoria,dhs2131,js1353\}@columbia.edu}}}

\maketitle

\begin{abstract}
  Assignment markets involve matching with transfers, as in labor markets and housing markets. We consider a two-sided assignment market with agent types and stochastic structure similar to models used in empirical studies, and characterize the size of the core in such markets. Each agent has a randomly drawn productivity with respect to each type of agent on the other side. The value generated from a match between a pair of agents is the sum of the two productivity terms, each of which depends only on the type but not the identity of one of the agents, and a third deterministic term driven by the pair of types. We allow the number of agents to grow, keeping the number of agent types fixed. Let $n$ be the number of agents and $K$ be the number of types on the side of the market with more types.
  We find, under reasonable assumptions, that the relative variation in utility per agent over core outcomes  is bounded as $O^*(1/n^{1/K})$, where polylogarithmic factors have been suppressed. Further, we show that this bound is tight in worst case. We also provide a tighter bound under more restrictive assumptions.  Our results provide partial justification for the typical assumption of a unique core outcome in empirical studies.
\blankline

\end{abstract}

\end{titlepage}

\section{Introduction}
\label{sec:intro}

We study bilateral matching markets such as marriage markets, labor markets, and housing markets, that allow participants to form partnerships with each other for mutual benefit. The two classical
models of such matching markets are the {\em non-transferable} utility (NTU) model of Gale and Shapley~\cite{Gale1962}, where payments are not allowed between the agents; and the Shapley-Shubik-Becker {\em transferable} utility (TU)
model \cite{Shapley1971,becker1973theory}, where transfer payments are allowed between pairs of agents who form a match. For each of these models the natural solution concept is that of a {\em stable} outcome, in which there is no pair of agents who would be happier with each other than in their current outcome.
In fact, for TU matching markets, it is well known that the notion of a
stable outcome coincides with that of a competitive equilibrium. A stable outcome
is guaranteed to exist in any two-sided market, but is typically not unique.
The concept of stability is widely used as a starting point in theoretical and empirical studies
in the context of matching. A nearly unique
stable outcome is required in order to facilitate predictions, comparative statics and so on, but little is known about when this occurs in the TU
setting.\footnote{\label{note1}A small core has been found in special cases of the TU setting
as in \cite{goz1,goz2,HR2014}, which we discuss below. In the case of the NTU setting,
real markets have almost always been found to contain
a nearly unique stable outcome, e.g. \cite{RothPeranson}, and a body of theory explains this, e.g. \cite{ImmorlicaMahdian2005,KojimaPathak2009,ALK2013,Samet,LeshnoAzavedo}.}
In this work, we seek to characterize the size of the set of stable matches
as a function of market characteristics in TU matching markets.



The motivation for our work is twofold. First, uniqueness of the stable outcome is typically assumed in empirical investigations, though there is insufficient theoretical basis to justify such an assumption. We ask when such an assumption is justified. Second, it is of interest to know whether basic market primitives, i.e., the number of agents and the values of possible matches, are sufficient to determine the outcome of the market, or whether there is significant ambiguity arising from which equilibrium the market is in. Can a labor market support higher wages for labor without adding jobs or improving productivity, just by moving to a different equilibrium?
In TU matching markets, market primitives like the value generated by a pair/match, and even transfers occurring in outcomes are difficult to observe, which has hindered empirical studies of features like core size (NTU markets are much easier to study empirically\footnote{ See footnote~\ref{note1}.}). This further increases the importance of generating theoretical predictions of core size, which can also potentially guide future empirical work.

We consider the assignment game model of Shapley and Shubik~\cite{Shapley1971}, consisting of
``workers" and ``firms" each of whom can match with at most one agent on the other side.  
To model the different skills of the workers and the different requirements of the firms, we assume that there are $K$ types of workers and $Q$ types of firms.
Matching worker $i$ with firm $j$ generates a value $\Phi_{ij}$ (this can be divided between $i$ and $j$ in an arbitrary manner since transfers are allowed), which we model as a sum of two terms: a term $u(\cdot, \cdot)$ that depends only on the {\em types} of $i$ and $j$, and a term $\psi_{i,j}$ that represents the ``idiosyncratic'' contributions of worker $i$ to firm $j$. 
In our model the $u(\cdot,\cdot)$ is assumed to be fixed, but the $\psi_{ij}$ is the sum of two random variables,
the ``productivity" of worker $i$ with respect to the \emph{type} of firm $j$ and, symmetrically,
the ``productivity" of firm $j$ with respect to the \emph{type} of worker $i$. These
productivities are assumed to be independently drawn from a bounded distribution (satisfying certain assumptions) for each (agent,type) pair.
In addition to being normatively attractive, such a generative model for the value of a match
has been used in empirical studies of marriage markets, starting with Choo and Siow \cite{Choo2006,chiappori2011partner,Galichon2010}.

We study the size of the set of stable outcomes for a random market constructed in this way.
Shapley and Shubik~\cite{Shapley1971} showed that the set of stable outcomes (which is the same as the core)
has a lattice structure, and thus has two extreme stable matchings: the worker optimal stable match, where each worker earns
the maximum possible and each firm the minimum possible in any stable matching; and
the firm optimal stable matching which is the symmetric counterpart. Also,
all stable outcomes live on a maximum weight matching, which is generically unique.
Given these structural properties, our metric for the size of the core is quite natural:
we consider the difference between the maximum and minimum utility of a worker (equivalently, a firm) in the core, averaged over matched workers (or firms).
Our main result is that the size of the set of stable matchings,
as measured by this metric, is small under some reasonable assumptions
on market structure: specifically, the expected core size is
$O^*(1/\sqrt[\ell]{n})$ in a problem with $n$ agents, and at most $\ell$ types
of agents on each side (with $\ell$ fixed). We show that this bound is essentially tight by
constructing a sequence of markets such that the core size is $\Omega(1/\sqrt[\ell]{n})$.
Thus the core shrinks with market size, and this shrinking is faster when there are fewer
types of agents. Additionally, we obtain a tighter upper bound in the special case with just one type of employer and more employers than workers. Our upper bound in this case improves sharply as the number of additional employers $m$ increases; we establish a bound of $O^*(1/(n^{1/\ell}m^{1-1/\ell}))$, where $\ell$ is the number of worker types.

Our model has the following property (here, think of $u(\, \cdot\, , \, \cdot \,)$
as being formally incorporated in the worker productivity): For every (worker type, firm type) pair,
there is a ``price" associated with this type-pair, such that
for every matched pair of agents of these types, the utility of each agent is her
productivity (with respect to the type on the other side),
``corrected" additively (in opposite directions) by the price. We show
that variation in these type-pair prices is uniformly bounded as
$O^*(1/\sqrt[\ell]{n})$ across core allocations, in expectation, implying the bound on core size.
A key component of our analysis is to relate the combinatorial structure of the core to
 order statistics of certain independent identically distributed (i.i.d.) random variables (r.v.s).
These r.v.s are one-dimensional projections of point processes in (particular subregions of) the unit hypercube,
where the point processes correspond to the market realization.
An analytical challenge that we face is that the relevant projections as well as the
relevant order statistics are themselves a random function of the market realization.
We overcome this via appropriate union bounds. Our analysis throws light on which
aspects of market structure affect the core and its size.

Most of the related literature focuses on the NTU model of Gale and Shapley~\cite{Gale1962}. For that model, a number of papers establish a small core under various assumptions such as short preference lists~\cite{ImmorlicaMahdian2005,  KojimaPathak2009, KojimaPathakRoth2013},
strongly correlated preferences~\cite{Samet,LeshnoAzavedo}.
In a recent paper Ashlagi et al.~\cite{ALK2013} show
that in a random NTU matching market with long lists and uncorrelated preferences,
even a slight imbalance results in a significant advantage for the short side of the market
and that there is approximately a unique stable matching. Further, the near uniqueness of the stable matching
is found to be robust to varying correlations in preferences and other features, suggesting
that a small core may be generic in NTU matching markets.
There is an extensive literature on large
assignment games that  extends the many structural properties established by Shapley and Shubik
for finite assignment games to a setting in which the agents form a continuum, see for example
Gretzky, Ostroy and Zame~\cite{goz1,goz2}. Those papers also show convergence of large finite
markets to the continuum limit, including that the core shrinks to a point.
However, unlike in our model, they model the productivity of each partnership
as a deterministic function of the pair of types, with the only randomness being in the number of agents of each type. 
The work on assignment games that is most closely related to our work
is a recent preprint of Hassidim and Romm~\cite{HR2014}: 
in their model, all workers (firms) are a priori identical,
and the value of matching worker $i$ to firm $j$ is a random draw from a bounded
distribution, independently for every pair $(i,j)$.
For such a model, they establish an approximate ``law of one price,''
i.e., that workers are paid approximately identical
salaries in any core allocation, and that the long side gets
almost none of the surplus in unbalanced markets. In contrast, we work
with multiple types of workers and firms, and the value of
a match depends on the types of each agent, and
random variables that depend on the identity of
one of the agents and the {\em type} (but not the identity) of the other agent.

The rest of the paper is organized as follows. We present our model in Section \ref{sec:model}, our results in Section \ref{sec:results}, and an overview of the proof of our main result in Section \ref{sec:proofs}. We conclude with a discussion in
Section \ref{sec:discussion}. Several proofs are deferred to appendices.

\section{Model Formulation}
\label{sec:model}

We consider a two-sided, transferable utility matching market with a finite number of agents. The sides of the market are represented by the labor ($\ml$) and the employers ($\me$). Let $n_{\ml}$ be the number of agents in $\ml$ and $n_{\me}$ be the number of agents in $\me$; we let $n := |\ml| + |\me|$ denote the size of the market, i.e., the total
number of agents in the problem.
We assume that the underlying graph is complete, that is, all pairs of agents can potentially be matched.
Each side of the market is partitioned into a finite number of {\em types}
and we let $K$ and $Q$ denote the number of different types of agents in $\ml$ and $\me$ respectively. We define $\mathcal{T}_{\ml} := \{1, \ldots, K\}$ and $\mathcal{T}_{\me} := \{1, \ldots, Q\}$ to be the set of types in the labor and employer side respectively. Let $\mathcal{T}=\mathcal{T}_{\ml}\times\mathcal{T}_{\me}$ denote the set of pairs of types.
If $n_{\ml}=n_{\me}$ we say that the problem is \emph{balanced}. Otherwise, we say that the problem is \emph{unbalanced}. In addition, for a given type $t \in \typelab\cup \typeemp$, we denote by $n_t$ the number of agents of type $t$.
Finally, let $\tau(a)$ denote the type of agent $a \in \ml \cup \me$; given a type $t$ and an agent $a$,  we say that $a\in t$ if $\tau(a)=t$. In what follows we typically use $i$ to denote an individual agent in $\ml$, and $j$ to denote an individual agent in $\me$.

The value of the match between $i$ and $j$ is denoted $\Phi(i,j)$.
An \emph{outcome} is a pair $(M,\gammabf)$, where $M$ is a matching between agents in $\ml$ and $\me$, and $\gammabf$ is a payoff vector such that $\gamma_i + \gamma_j \; = \; \Phi(i,j)$ for every pair of matched agents $i\in \ml$, $j \in \me$, $(i,j) \in M$. That is, the vector $\gammabf$ indicates how the value of a match is divided among the agents involved in the match.
In this paper we shall be concerned with outcomes that are in the {\em core}, i.e., outcomes such that no coalition of players can produce greater value among themselves than the sum of their utilities. Shapley and Shubik~\cite{Shapley1971} show that for this matching market model, an outcome $(M,\gammabf)$ is in the core if and only if it is satisfies stability.
The stability condition requires
$\gamma_i + \gamma_j \geq \Phi(i,j)$ for all $i \in \ml$ and $j \in \me$, and further requires the $\gammabf$ vector to
be non-negative.\footnote{Note that in any unstable outcome, there must either be an individual agent who would prefer to not participate in the matching (because of a negative payoff) or a \emph{blocking pair} of agents who can both do better by matching with
each other (because the value they generate by matching with each other exceeds their current payoffs).} 
The set of stable outcome utilities turns out to be the set of optima of the dual to the maximum weight matching linear program, implying in particular that the matching $M$ in a stable outcome must be a maximum weight matching.

\subsection{Structure of $\Phi(i,j)$}

We assume that $\Phi(i,j)$ is additively separable as follows.

\begin{assumption*}[Separability]
$\Phi(i,j) = u(\tau(i),\tau(j)) + \epsilon^{\tau(i)}_j + \eta^{\tau(j)}_i.$
\end{assumption*}
\noindent

It is natural to think that the value of matching $i$ and $j$ can be broken down into a sum of two components: a utility $u(\tau(i),\tau(j))$ that depends only on the agents' types, and a term $\psi^{\tau(i),\tau(j)}_{i,j}$ which is match specific and potentially depends on both the identity of the agents as well as their types. The separability assumption states that the match-specific component is further additively separable into two terms that each depend on the identity of one of the agents and only the type of the other agent. In particular, for any fixed employer $j$ and two distinct workers $i,i' \in \ml$ we have $\epsilon^{\tau(i)}_j=\epsilon^{\tau(i')}_j$ whenever $\tau(i)=\tau(i')$, as the term $\epsilon$ only depends on the \emph{type} of the agents in $\ml$. Analogously, the term $\eta$ depends on the \emph{individual} worker $i\in \ml$ but only the \emph{type} of the firm $j\in \me$.

We model the term $u(\tau(i), \tau(j))$ as a fixed constant, whereas the $\epsilon$ and $\eta$ terms are modelled as random variables, independent across agent type pairs. The continuum limit of such a model was introduced by Choo and Siow \cite{Choo2006}, who used the model to empirically estimate certain structural features of marriage markets. Such a model is attractive in allowing for reasonable heterogeneity and idiosyncratic variation via the random variables, while still remaining structured due to a fixed number of types. While these features have been important in facilitating estimation~\cite{Choo2006,chiappori2011partner}, they simultaneously also make this a plausible model of real markets.

We further assume that the terms $\epsilon^{\tau(i)}_j, \eta^{\tau(j)}_i$ are independent random draws from the
uniform $[0,1]$ distribution. While the assumption of i.i.d. $U[0,1]$ r.v.s appears quite restrictive,
our results and proofs extend to arbitrary non-atomic bounded distributions supported on a closed interval,
with positive density everywhere in the support.

\subsection{Preliminaries}

We now state some preliminary observations on the structure of the core under the separability assumption. 
We start by showing that the payoffs can be expressed more conveniently.
For each $i\in \ml$ and each type $q\in \typeemp$, let $\tilde{\eta}^q_i = u(\tau(i), q) + \eta^q_i$.

In our market model with probability 1 the maximum weight matching is unique, 
so we assume a unique maximum weight matching $M$ to simplify the exposition.
We denote by $M(t)$ the set of agents who are matched to an agent of type $t$ under $M$.
In addition, we use $U$ to denote the set of unmatched agents under matching $M$. 

\begin{proposition}\label{prop:introducing_alphas}
Let $M$ be the unique maximum weight matching. Any core solution $(M,\gammabf)$, corresponds to a vector $\alphabf \in \real^{K \times Q}$ such that the payoffs can be expressed as:
\begin{itemize}
\item $\gamma_i = \tilde{\eta}^{q}_i - \alpha_{kq}$, for all $i\in \ml$ such that $ \tau(i)=k$ and $i \in M(q)$.
\item $\gamma_j = \epsilon^k_j + \alpha_{kq}$, for all $j \in \me$ such that $\tau(j)=q$ and $j \in M(k)$.
\end{itemize}
\end{proposition}

Proposition~\ref{prop:introducing_alphas} follows from stability, and formalizes the existence of a single ``price" for every type-pair $(k,q)$ that is common across all matched pairs of agents with those types.
Based on Proposition~\ref{prop:introducing_alphas}, any core solution can be expressed in terms of the maximum weight matching $M$ and the vector $\alphabf$.

The following proposition states necessary and sufficient conditions for $(M,\alphabf)$ to be a core outcome. (The maximum over an empty set is defined as $-\infty$.)
\begin{proposition}\label{prop:general_nec_cond}
Let $M$ be the unique maximum weight matching. The following conditions are necessary and sufficient for $(M,\alphabf)$ to be a core solution:
\begin{enumerate}
\item[(ST)] For every pair of types $(k,q),~(k',q') \in \mathcal{T}$:
$$\min_{i \in k'\cap M(q')} \tilde{\eta}^{q'}_i-\tilde{\eta}^{q}_i + \min_{j \in q\cap M(k)}\epsilon^k_j - \epsilon^{k'}_j \geq \alpha_{k'q'} - \alpha_{kq} \geq \max_{i \in k\cap M(q)} \tilde{\eta}^{q'}_i-\tilde{\eta}^{q}_i + \max_{j \in q'\cap M(k')}\epsilon^k_j - \epsilon^{k'}_j.$$
\item[(IM)]For every pair of types $(k,q) \in \mathcal{T}$:
\begin{align*}
\phantom{\textup{and} \quad} \min_{j \in q\cap M(k)}\epsilon^k_j  &\geq  - \alpha_{kq} \geq  \max_{j \in q \cap U}\epsilon^k_j\,, \qquad\\
\textup{and} \quad \min_{i \in k\cap M(q)}\tilde{\eta}^q_i  &\geq  \; \, \alpha_{kq}\ \geq  \max_{i \in k \cap U}\tilde{\eta}^q_i\,.
\end{align*}
\end{enumerate}

\end{proposition}

The first set of conditions follow from the non-existence of a blocking pair of matched agents. The second conditions follow from the fact that utilities are non-negative (implying the left inequalities) and the non-existence of a blocking pair involving an unmatched agent. See \cite[Proposition 1]{chiappori2011partner} for a proof.

We conclude with a definition of the size of the core, denoted by $\core$.
We define $\core$ as the difference between the maximum and minimum utility of a worker (or firm) in the core, averaged over workers matched under $M$. This can be equivalently stated in terms of the vector $\alphabf$. 
For each pair of types $(k,q) \in \mathcal{T}$, let $\alpha^{\max}_{kq}$ and $\alpha^{\min}_{kq}$ be the maximum and minimum possible values of $\alpha_{kq}$ among core $\alpha$ vectors.

\begin{definition}[Size of the core]
Let $M$ be the unique maximum weight matching. For each pair of types $(k,q)\in \mathcal{T}$, let $N(k,q)$ denote the number of matches between agents of type $k$ and agents of type $q$. Then, the size of the core is denoted by $\core$ and is defined as:
$$\core = \frac{\sum_{k} \sum_{q} N(k,q)|\alpha^{\max}_{kq} - \alpha^{\min}_{kq}|}{\sum_{k} \sum_{q} N(k,q)}.$$
\end{definition} 

\section{Results}\label{sec:results}

We keep the number of agent types 
fixed and allow the number of agents to grow, focusing on how the size of the core
scales as the market grows.

Given the stochastic nature of the our problem, the size of the core $\core$
is itself a random variable. Therefore, the main focus of our work is to study
how the expected value of $\core$ depends on the characteristics of the market.
In finite markets it is generically possible to marginally
modify some payoffs in a core solution without violating stability and, therefore,
the size of the core is strictly positive \cite{Shapley1971}. However, as the size of the market
increases (the agent types stay the same), the set of core vectors $\alpha$ should shrink as
an increase in the number of
stability constraints limits the possible perturbations to the payoffs, cf.
Proposition \ref{prop:general_nec_cond}. 

We start by considering the simple case of markets with one type on each side,
that is $K=Q=1$. Given that there is only one type of agent on each side, the deterministic
utility term $u=u(\tau(i),\tau(j))$ will be the same for all possible matches,
regardless the identity of the agents.
The value of a match between agents $i \in \ml$ and $j\in \me$ is
$\Phi(i,j) = u + \eta_i + \epsilon_j$. Suppose $u>0$.

\begin{remark}
In the case of a balanced market, i.e., $n_{\ml} = n_{\me}$, the above market
has $\core \geq u$ with probability 1. In particular, $\E[\core] = \Omega(1)$.
\end{remark}

%
%
%
%
%
The idea is the following: all agents will be matched in a stable solution and by
Proposition~\ref{prop:introducing_alphas}, we can describe the size of the core in
terms of a single parameter $\alpha$; by Proposition~\ref{prop:general_nec_cond},
the core consists of all $\alpha \in  [- \min_{j} \epsilon_j, u+\min_{i}{\eta}_i]$.
In other words, the value $u$ that is part of $\Phi(i,j)$ for each $(i,j)$ can
be split in an arbitrary fashion between employers and workers.
On the other hand, in case of any imbalance, i.e., $n_{\ml} \neq n_{\me}$,
it turns out that $u$ must go entirely to the short side of the market, and
the size of the core is $O(1/n)$ (the distance between consecutive order statistics
of the $\eps_j$'s or the $\eta_i$'s). Thus, the core is small and rapidly shrinking
in any unbalanced market in the case of $K=Q=1$.


%


We now consider the general case of $K$ types of labor and $Q$ types of employers.
The following condition generalizes the imbalance condition to the case of multiple types.
The idea is to get rid of the cases that, for certain values of deterministic utilities $u(\, \cdot \, , \, \cdot \,)$, may resemble a balanced problem.
\begin{assumption}\label{assumpt:avoid_symmetry}
For ever pair of subsets of types $\mathcal{S}\subseteq \typelab$ and $\mathcal{S'}\subseteq \typeemp$ we must have $\sum_{t\in \mathcal{S}} n_t \neq \sum_{t \in\mathcal{S'}} n_t$. In words, this means that there is no subset of types such that the submarket formed by agents of those types is balanced.
\end{assumption}
We highlight that in our setting with fixed $K$ and $Q$ and growing $n$, ``most" markets
satisfy Assumption \ref{assumpt:avoid_symmetry}.\footnote{Consider possible vectors
$\cN = \{(n_{t})_{t \in \typelab \cup \typeemp }: \sum_t n_t =n\}$ describing the number of agents of each type.
Then $O(1/n)$ fraction of these vectors violate Assumption \ref{assumpt:avoid_symmetry}.}
We make a further regularity assumption, namely that the number of agents of
each type grows linearly in the size of the market.

\begin{assumption}\label{assump:essential_assumpts}
There exists $C > 0$ such that
for all types $t \in \typelab \cup \typeemp$,  we have $n_t \geq Cn$.
\end{assumption}

We now present our main theorem.
\begin{theorem}\label{thm:general_k_q}
Consider $K\geq 1$ types of labor, and $Q\geq 1$ types of employers. There exists
$f(n) = O^*\left(\frac{1}{\sqrt[\max(K,Q)]{n}}\right)$ such that under
Assumption~\ref{assumpt:avoid_symmetry} and Assumption~\ref{assump:essential_assumpts}, for
a market with $n$ agents
we have $\E[\core] \leq f(n)$. Further, there
exists a sequence of markets with $K$ types of labor and $Q$ types of employers
such that $\E[\core] = \Omega\left(\frac{1}{\sqrt[\max(K,Q)]{n}}\right)$.
\end{theorem}
In words, our main result says that under reasonable conditions,
$\E[\core]$ is vanishing as $n \rightarrow \infty$,
at a rate $O^*\left(\frac{1}{\sqrt[\max(K,Q)]{n}}\right)$ and that this bound is tight in worst case.
Thus, the core size shrinks to zero as the market grows larger, at a rate that
is faster (in worst case) if there are fewer types of agents. 
We give a proof of our main theorem in Section \ref{sec:overview_pf}, along with Appendices \ref{app:proofKQ} and \ref{app:lower_bound}.

The upper bound in Theorem~\ref{thm:general_k_q} can be improved
if further constraints are imposed on the number of types and the imbalance.
As an illuminating example, we show that in the setting in which $K \geq 2$,
$Q=1$ and $n_{\me} > n_{\ml}$, the size of the core can be bounded above by
a function that depends on both the size of the market and on the size of
the imbalance in the market.
\begin{theorem}\label{thm:K-1-emp}
Consider the setting in which $K\geq 2$, $Q=1$, $n_{\me} > n_{\ml}$ and let $m = n_{\me} - n_{\ml}$. 
Under Assumption~\ref{assump:essential_assumpts}, we have $\E[\core] \leq O^*\left(\frac{1}{n^{\frac{1}{K}}m^{\frac{K-1}{K}}}\right)$.
\end{theorem}
For $m=O^*(1)$, the bound in Theorem \ref{thm:K-1-emp} matches that in Theorem \ref{thm:general_k_q}. 
However, the bound here becomes tighter as the imbalance $m$ grows. In fact, for $m = \Theta(n)$,
the core size is bounded as $O^*(1/n)$. It is noteworthy that the scaling behavior here does not depend on the number of
worker types. We also mention here that, using symmetry, an analogous result can be stated
with $Q$ types of employers, only one type of worker, and more workers than employers.

We prove Theorem \ref{thm:K-1-emp} in Appendix \ref{app:K-1}. The idea is to use the unmatched agents 
and condition (IM) in Proposition \ref{prop:general_nec_cond} (for the employers)
to control absolute variation in one of
the $\alpha$'s. We separately control the relative variation of the $\alpha$'s in the core
using condition (ST) in Proposition \ref{prop:general_nec_cond} under Assumption~\ref{assump:essential_assumpts}.
Combining these we obtain the stated bound on $\core$.

\section{Overview of the proof of the main result}
\label{sec:proofs}

We now present an overview of our proof of Theorem~\ref{thm:general_k_q}.
We first discuss the key steps in establishing the upper bound (the complete proof can be found in Appendix~\ref{app:proofKQ}), and then sketch the proof of the lower bound in Section \ref{subsec:lb_proof_overview} (completed in Appendix~\ref{app:lower_bound}).

Throughout this section, we assume that there is a unique maximum weight matching and we refer to it as $M$. Given $M$, recall that $N(k,q)$ is defined as the number of matches between agents of type $k$ and agents of type $q$ in $M$. %
%

We start by constructing a graph associated with matching $M$ as follows. Let $G(M)$ be the bipartite graph whose vertex sets are the types in $\ml$ and $\me$, and such that there is an edge between types $k\in \typelab$, $q\in \typeemp$ if and only if there is an agent of type $k$ matched to an agent of type $q$ in $M$, i.e., $N(k,q)>0$. The following lemma states a key fact regarding the structure of $G(M)$.

\begin{lemma}\label{lemma:type_adjacency_graph}
Let $M$ be the unique maximum weight matching and let $G(M)$ be the associated type-adjacency graph. Suppose we \emph{mark} the vertex in $G(M)$ corresponding to type $t$ if and only if at least one agent of type $t$ is unmatched under $M$. Then, under Assumption~\ref{assumpt:avoid_symmetry}, with probability $1$, every connected component in $G(M)$ must contain a marked vertex.
\end{lemma}


%
%
%
%

\subsection{Overview of the upper bound proof}
\label{sec:overview_pf}

Roughly, the idea of the upper bound proof of Theorem~\ref{thm:general_k_q} is as follows. We consider some suitably defined events (which are discussed later), which occur in typical markets. Under these events, we show that the variation in the type-pair prices is uniformly bounded as follows,
$$ \max_{(k,q) \in \typelab \times \typeemp, N(k,q)>0}|\alpha^{\max}_{kq} - \alpha^{\min}_{kq}| \leq f(n),$$ for some $f(n)=O^*\left( \frac{1}{n^{1/\max(K,Q)}}\right)$. To prove this bound, we use the graph $G(M)$ as defined above. Given a type $t\in \typelab \cup \typeemp$, let the \emph{distance} $d(t)$ be defined as the minimum distance in $G(M)$ from $t$ to \emph{any} marked vertex. By Lemma~\ref{lemma:type_adjacency_graph}, every unmarked vertex $t$ must be at a finite distance from a marked one. Furthermore, $\max_{t\in \typelab \cup \typeemp} d(t) \leq K+Q$ regardless the realization of the graph.

Our argument to control the variation in the $\alpha$'s is by induction on $d(t)$. To establish our induction base, we show that the variation in \emph{all} the relevant $\alpha$'s associated with marked types (these types have distance zero) is bounded. 
In particular, for each marked type $t$, we show that $\max_{t':~N(t,t')>0} \left( \alpha^{\max}_{t,t'} - \alpha^{\min}_{t,t'}\right)\leq O^*\left( \frac{1}{n^{1/\max(K,Q)}}\right)$. This is done in Lemma~\ref{lemma:artificial_unmatched}.

In the inductive step, we assume the bound holds for every $\alpha$ associated with a type whose distance is $d$ or less, i.e, for every $(t,t') \in \typelab \times \typeemp$ such that $\min(d(t), d(t'))\leq d$, we have  $\alpha^{\max}_{t,t'} - \alpha^{\min}_{t,t'} \leq O^*\left( \frac{1}{n^{1/\max(K,Q)}}\right)$.  Then, we use the inductive hypothesis to show that the result must also hold for all types whose distance is $d+1$. By the definition of distance, for every type $t$ such that $d(t)=d+1$, there must exist a type $t^*$ such that $d(t^*)=d$ and $N(t,t^*)>0$. Therefore, by our inductive hypothesis, we must have  $\alpha^{\max}_{t,t^*} - \alpha^{\min}_{t,t^*} \leq O^*\left( \frac{1}{n^{1/\max(K,Q)}}\right)$. Using this bound, we further bound the variation in all $\alpha$'s associated with type $t$, by controlling the relative variation of the $\alpha$'s in the core, i.e., by showing that $\alpha_{t,t_1}-\alpha_{t,t_2}$ for types $t_1, t_2$ with matches to type $t$ can vary only within a range bounded by $O^*\left( \frac{1}{n^{1/\max(K,Q)}}\right)$. This is formally achieved in Lemma~\ref{lemma:all_matched}.

%

To conclude, we briefly describe the nature of the events that we argue must hold with high probability. These events are related to the distance between order statistics of the projections of points distributed independently in (sub-regions of) a hypercube. Note that, once we focus on a single type $t$, the random productivities associated to an agent of type $t$ can be described by a $D(t)$-dimensional vector within the $[0,1]^{D(t)}$-hypercube, where $D(t)$ is dimension of the productivity vector of agents of type $t$ (i.e., $D(t)=K$ if $t\in \typeemp$ and $D(t)=Q$ otherwise). Furthermore, the location of these points can be described by a point process in $[0,1]^{D(t)}$. Hence, all the conditions in Proposition~\ref{prop:general_nec_cond} can be interpreted as geometric conditions in the unitary hypercube. We use this geometric interpretation and relate Proposition~\ref{prop:general_nec_cond} to the regions, random sets and random variables defined below in Section \ref{subsubsec:hypercubemain} to prove our main theorem.

\subsection{Hypercube definitions and key lemmas}
\label{subsec:hypercube_and_keylemmas}

As mentioned in Section \ref{sec:intro}, a key component of our analysis is to relate the combinatorial structure of the core to order statistics of certain independent identically distributed (i.i.d.) random variables. These random variables are one-dimensional projections of point processes in (particular subregions of) the unit hypercube, where the point processes correspond to the market realization. Next, we formally define the regions, random sets and random variables that will be useful in our analysis.

\subsubsection{Hypercube definitions}
\label{subsubsec:hypercubemain}

Consider a type $t \in \typeemp$. For each employer $j: \tau(j) =t$,
there is a vector of productivities $\eps_j$ distributed uniformly in $[0,1]^K$,
independently across employers. In this subsection we consider these productivities
for a given $t$. We suppress $t$ in the definitions to simplify notation (so $n$ here
corresponds to $n_t$, and so on). Analogous definitions can be made for $t \in \typelab$.

Consider $n$ i.i.d. points $(\eps_j)_{j=1}^n$, distributed uniformly in the $[0,1]^K$-hypercube. Here $\eps_j=(\eps_j^1, \eps_j^2, \ldots, \eps_j^K)$. Let $\cK = \{1,2, \ldots, K\}$ denote the set of dimension indices. Define the region
\vspace{-3pt}
\begin{align}\label{def:R_k}
  \cR^k = \{ x \in [0,1]^K: x^k \geq x^{k'} \,\quad \forall k' \neq k, k' \in \cK\}
\end{align}

For $k_1, k_2 \in \cK, k_1\neq k_2$ and for $\delta \in [0,1/2]$, define the region
\begin{align}\label{def:R_k1_k2}
  \cR^{k_1,k_2}(\delta) = \{ x \in [0,1]^K: x^{k_1} \geq x^{k} \, \quad  \forall k \notin \{k_1, k_2\}, k \in \cK, x^{k_1}\geq \delta\}\, .
\end{align}
\vspace{-4pt}
Let
\vspace{-3pt}
\begin{align}
  \cV^k \,&= \{x:x=\eps_j^k \textup{ for } \{j: \eps_j \in \cR^k \}\} \, , \label{eq:cVkdef}\\
  \text{and} \qquad V^k \,&= \max \big (\textup{Difference between consecutive values in
  $ \cV^k \cup \{0,1\}$} \big ) \, .\label{eq:Vkdef}
\end{align}
Thus, $\cV^k \subset [0,1]$ is the set of values of the $k$-th coordinate of the points lying in $\cR^k$, and $V^k \in \real$ is
the maximum difference between consecutive values in $\cV^k \cup \{0,1\}$. (As an example, if $\cV^k=\{0.3, 0.4,0.8\}$, the differences between consecutive values in $\cV^k \cup \{0,1\}$ are $0.3, 0.1, 0.4, 0.2$, resulting in $V^k = 0.4$.) Note that $\cV^k$ is a random and finite set, and $V^k$ is a random variable.
Let
\begin{align}
  \cV^{k_1, k_2}(\delta)\,&=\{x:x=\eps_j^{k_1} - \eps_j^{k_2} \textup{ for } \{j: \eps_j \in \cR^{k_1, k_2}\} \}\, , \label{eq:cVk1k2def}\\
\text{and} \qquad    V^{k_1, k_2}(\delta) \,&= \max \big (\textup{Difference between consecutive values in
   $\cV^{k_1,k_2}(\delta)\cup \{-1+\delta,1\}$} \big )   \label{eq:Vk1k2def}
\end{align}
Thus, $\cV^{k_1,k_2} \subset [-1+\delta,1]$ is the set of values of the difference between the $k_1$-th and $k_2$-th coordinate of points lying in $\cR^{k_1,k_2}$, and $V^{k_1,k_2} \in \real$
is the maximum difference between consecutive values in $\cV^{k_1,k_2} \cup \{-1+\delta,1\}$.

In addition, for $\delta \in (0,1/2]$ and $k \in \cK$, define
\begin{align}\label{def:cuboid}
  \tcR^k(\delta) = \{x\in [0,1]^K: x^{k'} \leq \delta  \, \quad \forall k' \in \cK, k' \neq k\}\, .
\end{align}
\vspace{-4pt}
Let
\vspace{-3pt}
\begin{align}
  \tcV^k(\delta) \,&= \{x:x=\eps_j^k \textup{ for } \{j: \eps_j \in \tcR^k \}\} \label{eq:cVcubdef}\\
\text{and} \qquad  \tV^k(\delta) \,&= \max \big (\textup{Difference between consecutive values in
  $ \tcV^k(\delta) \cup \{0,1\}$} \big ) \, .\label{eq:Vcubdef}
\end{align}

We now relate the above definitions to the combinatorial structure of our problem. We now include the type $t$ explicitly in the names of the associated regions, sets and random variables, e.g., region $\cR^k(\delta)$ when defined for type $t$ is referred to as $\cR^k(t,\delta)$.

The definition of these regions, sets and random variables might seem arbitrary at first sight. However, it is closely related to the geometric interpretation of the stability conditions. Intuitively, for a fixed type $t\in \typeemp$ with unmatched agents, one can bound  $\alpha_{kt}$ by using condition (IM) in Proposition~\ref{prop:general_nec_cond}: $\min_{j \in t\cap M(k)}\epsilon^k_j  \geq  - \alpha_{kt} \geq  \max_{j \in t\cap U}\epsilon^k_j.$ 
To apply this bound, we just care about the projection onto the $k$-th coordinate of the points $\epsilon_j$ with $j\in M(k)\cup U$. The main analytical challenge we face is that the these relevant subregions are themselves a random function of the market realization, as both $M(k)$ and $U$ are themselves random sets. We overcome this by appropriately defining the region $\tcR^k(t,\delta)$ so that it only contains points corresponding to agents in $M(k)\cup U$. 
Once we have done that, it should be easy to see that $\min_{j \in t\cap M(k)}\epsilon^k_j  - \max_{j \in t \cap U}\epsilon^k_j$ is upper bounded by the maximum distance between two consecutive points in $\tcR^k(t,\delta)$, when projected onto their $k$-th coordinate (the corner cases of all points being in $M(k)$, or in $U$, turn out to be easy to handle). This becomes precise once we introduce the set $\cV^k(t)$ and the random variable $V^k(t,\delta)$. Analogously, the regions $\cR^{k}(t)$ (for appropriate $k$) and  $\cR^{k_1,k_2}(t,\delta)$ (for appropriate $k_1,k_2$) allow us to apply the conditions (IM) and (ST) respectively, to bound the variation of $\alpha$'s associated with a type $t$. These relationships are more involved, so the explanation is delayed to the proofs.

Using the above notation, we now define the two events that will help us prove the results:
\vspace{-3pt}
\begin{align}\label{def:event_triangle}
    \Ev_1(t,\delta) = \Big \{ \max \Big ( \max_{k \in \cK} V^k(t), \max_{(k_1, k_2)\in \cK^{(2)}} V^{k_1, k_2}(t,\delta) \Big ) \leq f_1(n_t,K)\, \Big \},
\end{align}
\noindent for some $f_1(n_t,K) = O^*(1/n_t^{1/K})$ defined in Lemma~\ref{lemma:max_separation_hypercube}, $\delta \in [0,1/2]$ and where $\cK^{(2)}=\{(k_1,k_2):k_1, k_2 \in \cK, k_1\neq k_2\}$. (If $K=1$, then $\cK^{(2)}$ is the empty set $\emptyset$ in which case we follow the convention that $\max_{\emptyset } [\, \cdot \,]= -\infty$.). In addition,
\vspace{-3pt}
\begin{align}\label{def:event_cuboid}
 \Ev_2(t,\delta)= \Big \{ \max_{k \in \typelab} \tV^k(t,\delta) \leq f_2(n_t)/\delta^{K-1} \Big \}\
\end{align} 
\noindent for some $f_2(n_t) = O^*(1/n_t)$ defined in Lemma~\ref{lemma:cuboid_orderstat} and $\delta  \in (0,1]$.

The proof of all lemmas auxiliary to the proof of Theorem~\ref{thm:general_k_q} assume that these events (or some subset of them) occur. As shown by the next result (proved in Appendix \ref{app:hypercube_results}), that assumption does not pose a problem as these events simultaneously occur with high probability.

\begin{lemma}\label{lemma:all_events_prob}
There exists $\hat{C}=\hat{C}(K,Q)< \infty$ such that, for any $\delta = \delta(n) \in (0,1/2]$, the event $\bigcap_{t\in  \typelab \cup \typeemp} \left( \Ev_1(t,\delta)\cap \Ev_2(t,\delta) \right)$ occurs with probability at least $1-\hat{C}/n$.
\end{lemma}

\subsubsection{Statements of the key lemmas}
For every type $t\in \typelab \cup \typeemp$, we define $\vartheta(t)$ as $\vartheta(t) = \{k \in \typelab:~ N(k,t)>0 \}$ when $t \in \typeemp$ and  $\vartheta(t) = \{q \in \typeemp:~ N(t,q)>0 \}$ when  $t \in \typelab$. That is, $\vartheta(t)$ is the set of neighbours of $t$ in the graph $G(M)$.
Recall that, given a type $t\in \typelab \cup \typeemp$ we denote by $D(t)$ the \emph{dimension} of the productivity vector of agents of type $t$. That is, $D(t) = K$ if $t\in \typeemp$ and $D(t) = K$ if $t\in \typelab$.

\begin{lemma}\label{lemma:artificial_unmatched}
Consider the unique maximum weight matching $M$ and a type $t\in \typelab \cap \typeemp$. Let $\EvF_1(t)$ be the event
\vspace{-3pt}
\begin{align}\label{def:event_marked}
\EvF_1(t)=\{t\text{ is marked in $G(M)$ and at least one agent in $t$ is matched}\},
\end{align}
\noindent that is, $t$ has at least one unmatched and one matched agent. Let the events $\Ev_1(t,\delta)$ and $\Ev_2(t,\delta)$ be as defined by Eqs.~\eqref{def:event_triangle} and \eqref{def:event_cuboid} respectively. Under $\EvF_1(t) \cap \Ev_1(t,\delta) \cap \Ev_2(t,\delta)$, we have $$\max_{t'\in \vartheta(t)} \left( \alpha^{\max}_{t,t'} - \alpha^{\min}_{t,t'}\right) \leq \max\left(f_1(n_t,D(t))+ \delta, f_2(n_t)/\delta^{D(t)-1} \right),$$ 
\noindent where $f_1$ and $f_2$ agree with those in the definitions of events $\Ev_1(t,\delta)$ and $\Ev_2(t,\delta)$ respectively.
\end{lemma}

\begin{lemma}\label{lemma:all_matched}
Consider the unique maximum weight matching $M$ and a type $t\in \typelab \cap \typeemp$. Let $\EvF_2(t)$ be the event $$\EvF_2(t)=\{\text{all agents in $t$ are matched}\}.$$
Let the event $\Ev_1(t,\delta)$ be as defined by Eq.~\eqref{def:event_triangle}. Under $\EvF_2(t) \cap \Ev_1(t,\delta)$, for every $t^*\in \vartheta(t)$ we have $\max_{t'\in \vartheta(t)} \left( \alpha^{\max}_{t,t'} - \alpha^{\min}_{t,t'}\right) \leq \left(\alpha^{\max}_{t,t*}- \alpha^{\min}_{t,t*}\right) + 2f_1(n_t,D(t))+2\delta$, where $f_1$ agrees with the one in the definition of  $\Ev_1(t,\delta)$.
\end{lemma}
\tinyblank

Using the simple lemmas defined above, we provide a sketch of proof that, together with the explanation in Section~\ref{sec:overview_pf}, should suffice to roughly convey the idea while avoiding the technical details. As a reminder, the complete proof of the upper bound in Theorem~\ref{thm:general_k_q} can be found in Appendix~\ref{app:proofKQ}.

Let $n^*=\min_{t \in \typelab \cup \typeemp} n_t$ and let $\delta=1/{(n^*)}^{1/\max(K,Q)}$. Under Assumption~\ref{assump:essential_assumpts}, we have that $n^*=\Theta(n)$ and therefore $\delta=\Theta\big (1/{n}^{1/\max(K,Q)}\big )$. Furthermore, now $f_1(n_t,D(t))+ \delta$ and $f_2(n_t)/\delta^{D(t)-1}$, $2f_1(n_t,D(t))+ 2\delta$ as defined in the statements of Lemmas~\ref{lemma:artificial_unmatched} and~\ref{lemma:all_matched} are all $O^*\left(  \frac{1}{n^{1/\max(K,Q)}}\right)$. Using this choice of $\delta$ together with the inductive argument outlined in Section~\ref{sec:overview_pf}, we show that under the event 
$\bigcap_{t\in  \typelab \cup \typeemp} \left( \Ev_1(t,\delta)\cap \Ev_2(t,\delta) \right)$ we must have 
\vspace{-5pt}
$$ \max_{(k,q) \in \typelab \times \typeemp, N(k,q)>0}\left(\alpha^{\max}_{kq} - \alpha^{\min}_{kq}\right) \leq O^*\left( 1/n^{1/\max(K,Q)}\right).$$

\subsection{Proof of the lower bound}
\label{subsec:lb_proof_overview}

Our lower bound follows from the following proposition, proved in Appendix \ref{app:lower_bound}.

\begin{proposition}\label{prop:lower_bd}
Consider a sequence of markets (indexed by $\tn$) with $|\typelab|=K$ types of labor,
with $\tn$ workers of each type, and a single
type ``1" of employers, with $(K-1)\tn+1$ employers of this type.
(Assumptions \ref{assump:essential_assumpts} and \ref{assumpt:avoid_symmetry} are satisfied.)
Set $u(k_*,1)=0$ for
some $k_* \in \ml$, and $u(k,1)=3$ for all $k \in \ml \backslash k_*$.
For this market, we have $\E[\core]=\Omega^*(1/(n^{1/K}))$.
\end{proposition}

Note that the sequence of markets described can easily be ``dressed up" to fill in the
gaps in market sizes\footnote{Here $n=(2K-1)\tn+1$ for $\tn=1,2, \ldots$ but intermediate values of
$n$ can be handled by having slightly fewer workers of type $k_*$, which leaves
our analysis essentially unaffected.} and to accommodate
$Q \leq K$ types of firms\footnote{Let each worker type $q\neq 1$ have $\tn$ agents each
and $u(\, \cdot \, , q)=-2$. These workers are always unmatched, leaving the
core unaffected.}. If $Q> K$, we simply swap the roles of workers
and firms in our construction, leading to  $\E[\core]=\Omega^*(1/(n^{1/Q}))$ as needed.
Thus, the lower bound in Theorem \ref{thm:general_k_q} follows from Proposition \ref{prop:lower_bd}.

The rough intuition for our construction in Proposition \ref{prop:lower_bd} is as follows: For our choice of $u$'s it
is not hard to see that all workers of types different from $k_*$ are always matched in the core. One employer $j_*$ is matched to
a worker of type $k_*$.
Suppose vector $(\alpha_k)_{k \in \typelab}$ is in the core. Given that all types $k \neq k_*$ are a priori symmetric, we would expect that the $\alpha_k$'s for $k \neq k_*$ are close to each other (we formalize using Lemma \ref{lem:bd_min_diff_alpha_gen_k} that they are usually no more than $\delta \sim 1/\sqrt{\tn}$ apart). Assuming this is the case, we can order employers based on $X_j= \max_{k \neq k_*}\eps_j^k - \eps_j^{k_*}$, and $j_*$ should usually be the employer with smallest $X_j$, since this employer has the largest productivity with respect to $k_*$ relative to the other types. Now, the $X_j$'s are i.i.d., and a short calculation establishes that the distance between the first and second order statistics of $(X_j)_{j \in \me}$ is $\Theta(1/n^{1/K})$.
This ``large" gap between the first two order statistics allows for $(\alpha_{k_*}, (\alpha_k+\theta)_{k \neq k_*})$ to remain within the core for a range of values of $\theta \in \reals$ that has expected length $\Theta(1/n^{1/2})$ for $K=2$ and $\Theta(1/n^{1/K}) - \Theta(\delta) = \Theta(1/n^{1/K})$ for $K > 2$, leading to the stated lower bound on $\core$.

We remark that the key quantity here, the gap between the first two order statistics of $(X_j)_{j \in \me}$, is determined by the tail behavior (both the left and right tails) of the $\eps$'s, along with the number of types $K$. See Section \ref{sec:discussion} for further discussion. 

\section{Discussion}
\label{sec:discussion}

This paper quantifies the size of the core in matching markets with transfers,
as a function of market characteristics. We considered a model of an
assignment market with a fixed number of types of workers and firms.
We modelled the value
of a match between a pair of agents as a sum of a deterministic term determined by the
pair of types, and a random component which is the sum of two terms, each depending
on the identity of one of the agents and the type of the other.
Under reasonable assumptions, we showed that the size of the core
is bounded as $O^*(1/n^{1/\ell})$, where each side of the market
contains no more than $\ell$ types.

Our work answers some questions but raises several others. One
question is what happens if the random productivity
terms are drawn from unbounded distributions. For the market we construct
for our lower bound, the core size is determined by the tail behavior of the
random productivities, cf. Section \ref{subsec:lb_proof_overview}, suggesting that the
core could be larger in worst case if the productivities have an unbounded distribution.

On the other hand, it is of interest to understand the core in
typical/average case markets, as opposed to worst case markets.
Our bound of $O^*(1/n)$ for the special case of only one type of employer and
$\Theta(n)$ more employers than workers (a corollary of Theorem \ref{thm:K-1-emp})
does not depend on the number of worker types, in contrast to our general bound,
which implies that a relatively larger core can result in worst case from having more types.
How does the core size depend on the number of types in typical/average case markets?

It would be interesting to extend our results to many-to-one markets, where
employers can each have more than one opening. We expect that our results regarding
the core (also our proofs) extend to the case where 
each employer has capacity bounded by a constant, and 
employer utility is additive across matches. 

\newpage

\bibliographystyle{abbrv}
\bibliography{SSCore}{}

\appendix
    \setcounter{lemma}{0}
    \setcounter{claim}{0}
    \setcounter{remark}{0}
    \renewcommand{\thelemma}{\Alph{section}.\arabic{lemma}}
    \renewcommand{\theclaim}{\Alph{section}.\arabic{claim}}
    \renewcommand{\theremark}{\Alph{section}.\arabic{remark}}

\setcounter{lemma}{0}
\setcounter{claim}{0}
\setcounter{remark}{0}

\section{Results on point processes in the unit hypercube}
\label{app:hypercube_results}

Consider the $K$ dimensional unit hypercube $[0,1]^K$, and the Poisson process of uniform rate $n$
in this hypercube, leading to $N$ points $(\eps_i)_{i=1}^N$. (Note that $\E[N]=n$.)
Here $\eps_i=(\eps_i^1, \eps_i^2, \ldots, \eps_i^K)$. Let $\cK = \{1,2, \ldots, K\}$ denote the set
of dimension indices.

Let $\cR^k$ be the region defined by Eq.~\eqref{def:R_k}, and let $\cV^k$ and $V^k$ be as defined by Eqs.~\eqref{eq:cVkdef} and \eqref{eq:Vkdef} respectively. Similarly, let $\cR^{k_1,k_2}(\delta)$ be the region defined by Eq.~\eqref{def:R_k1_k2}, and let $\cV^{k_1, k_2}(\delta)$ and $V^{k_1, k_2}(\delta)$ be as defined by Eqs.~\eqref{eq:cVk1k2def} and ~\eqref{eq:Vk1k2def} respectively.

The following lemma, key to our proof of Theorem~\ref{thm:general_k_q}, says that with high probability, all the $(V^k)$'s and the $(V^{k_1, k_2})$'s are no larger than a (deterministic) function\footnote{In fact, our proof of Lemma \ref{lemma:max_separation_hypercube} identifies a bound of $(C\log n /n)^{1/K}$ where $C=6K(K-1)$, for sufficiently large $n$.} of $n$ that scales as $O^*(1/n^{1/K})$.

\begin{lemma}\label{lemma:max_separation_hypercube} Let $\cR^k$ be the region defined by Eq.~\eqref{def:R_k}, and let $\cV^k$ and $V^k$ be as defined by Eqs.~\eqref{eq:cVkdef} and \eqref{eq:Vkdef} respectively. Similarly, let $\cR^{k_1,k_2}(\delta)$ be the region defined by Eq.~\eqref{def:R_k1_k2}, and let $\cV^{k_1, k_2}(\delta)$ and $V^{k_1, k_2}(\delta)$ be as defined by Eqs.~\eqref{eq:cVk1k2def} and ~\eqref{eq:Vk1k2def} respectively.
  Fix $K \geq 1$. Then there exists $f(n,K) = O^*(1/n^{1/K})$ such that for any $\delta=\delta(n) \in [0,1/2]$ the following holds: Let
  \begin{align}\label{def:ev_1}
    \Ev_1 = \left \{ \max \left ( \max_{k \in \cK} V^k, \max_{(k_1, k_2)\in \cK^{(2)}} V^{k_1, k_2}(\delta) \right ) \leq f(n,K)\, \right \},
  \end{align}
  where $\cK^{(2)}=\{(k_1,k_2):k_1, k_2 \in \cK, k_1\neq k_2\}$. (If $K=1$, then $\cK^{(2)}$ is the empty set $\emptyset$ in which case we follow the convention that $\max_{\emptyset } [\, \cdot \,]= -\infty$.)
  We have
  $$
  \prob(\Ev_1) \geq 1-1/n \, .
  $$
\end{lemma}

\begin{proof}

Let $m=\lfloor 1/(C\log n/n)^{1/K}\rfloor$ for some $C< \infty$ that we will choose later, and let $\Delta = 1/m$. Note that
\begin{align}
\Delta \geq (C\log n/n)^{1/K} \, .
\label{eq:Delta_lb}
\end{align}
In our analysis of $V^k$ (resp. $V^{k_1,k_2}$), we will divide the interval $[0,1]$ (resp. $[-1+\delta,1]$) into subintervals of size $\Delta$ each, and show that with large probability, each subinterval contains at least one value of $\eps_i \in \cR^k$ (resp. $\eps_i^{k_1} - \eps_i^{k_2} \textup{ for } \{i: \eps_i \in \cR^{k_1, k_2}\}$). We will find that the density of points in $\cV^k$ (resp. $\cV^{k_1, k_2}$) is smallest near 0 (resp. $-1+\delta$), but even for the interval $[0,\Delta]$ (resp. $[-1+\delta, -1+\delta+\Delta]$), the number of points is Poisson with parameter $\Theta(n\Delta^K) = \Theta(\log n)$, allowing us to obtain the desired result for appropriately chosen $C$.

We first present our formal argument leading to a bound on $V^k$, followed by a similar argument leading to a bound on $V^{k_1, k_2}$.
Let
\begin{align}
  \Ev^k \equiv \bigcap_{i=0}^{m-1}\,\{\, [i\Delta, (i+1)\Delta]\cap \cV^k \neq \emptyset \, \}\, ,
  \label{eq:Evk_def}
\end{align}
where $\emptyset$ is the empty set.
Clearly, $\Ev^k \Rightarrow V^k \leq 2\Delta$. We now show that for any $k \in \cK$, we have $\prob\big(\xoverline{\Ev^k} \big) \leq 1/n^{K+2}$, for appropriately chosen $C$.
Define
\begin{align}
  h^j(x, \theta) =
   \left \{\begin{array}{ll}
  x^j & \textup{for } x\in [\theta,1]\\
  0 & \textup{otherwise}\, .
  \end{array} \right .
  \label{eq:fjdef}
\end{align}
It is easy to see that $\cV^k$ follows a Poisson process with density $nh^{K-1}(\, \cdot \, , 0)$. The number of points in interval $[i\Delta, (i+1)\Delta]$ is hence Poisson with parameter
\begin{align*}
 n\int_{i\Delta}^{(i+1)\Delta} h^{K-1}(x) \, \D x = ((i+1)^K -i^K)n\Delta^K/K \geq n\Delta^K/K \geq C\log n/K \, ,
\end{align*}
where we used the lower bound on $\Delta$ in \eqref{eq:Delta_lb}.
It follows that
\begin{align*}
  \prob([i\Delta, (i+1)\Delta]\cap \cV^k = \emptyset ) \leq \exp(-C\log n/K) =1/n^{C/K} \leq 1/n^3 \, ,
\end{align*}
for $C \geq 3K$.
We deduce by union bound over $i=0, 1, \ldots, m-1$ and De Morgan's law on \eqref{eq:Evk_def} that
\begin{align*}
  \prob\big(\xoverline{\Ev^k}\big) \leq m/n^{3} \leq n^{1/K}/n^{3} \leq 1/n^{2} \, .
\end{align*}
Using union bound over $k$ we deduce that
\begin{align}
\prob\Big (\cup_{k}\xoverline{\Ev^k}\Big ) \leq K/n^{2}
\label{eq:Ekc_bound}
\end{align}

We now present a similar argument to control $V^{k_1, k_2}$ when $K\geq 2$. Let $m'=(1-\delta)/\Delta$. (To simplify notation
we assume $m'$  is an integer. The case when it is not an integer can be easily handled as well.)
Let
\begin{align}
  \Ev^{k_1,k_2} \equiv \bigcap_{i=-m'}^{m-1}\,\{\, [i\Delta, (i+1)\Delta]\cap \cV^{k_1,k_2} \neq \emptyset \, \}\, ,
  \label{eq:Evk1k2_def}
\end{align}
where $\emptyset$ is the empty set.
Clearly, $\Ev^{k_1,k_2} \Rightarrow V^{k_1,k_2} \leq 2\Delta$. We now show that for any $k_1 \neq k_2$, we have 
$\prob\big(\xoverline{\Ev^{k_1,k_2}}\big) \leq K(K-1)/n^{2}$, for appropriately chosen $C$.
 It is easy to see that the two-dimensional projection $(x,y)=(\eps_i^{k_1}, \eps_i^{k_2})$ of points in
 $\cR^{k_1,k_2}$ follows a two-dimensional Poisson process with density $h^{K-2}(x)\ind(y \in [0,1])$, cf. \eqref{eq:fjdef}. We deduce that values in $\cV^k$ follow a one-dimensional Poisson process with density $ng$ for $g=h^{K-2}(\, \cdot\, ,\delta)*\ind(\in[-1,0])$, where $*$ is the convolution operator.
 A short calculation yields
 \begin{align*}
   g(x) =
   \left \{
   \begin{array}{ll}
   \big[ (x+1)^{K-1}-\delta^{K-1}\big ]/(K-1) & \textup{ for } x\in [-1+\delta,0)\\
   \big[ 1-\delta^{K-1}\big ]/(K-1) & \textup{ for } x\in [0, \delta)\\
   (1-x^{K-1})/(K-1) & \textup{ for } x \in [\delta,1]\\
   0 & \textup{ otherwise.}
   \end{array}\right .
 \end{align*}

The number of points in interval $[i\Delta, (i+1)\Delta]$ is Poisson with parameter
\begin{align*}
 n\int_{i\Delta}^{(i+1)\Delta} g(x) \, \D x  \, .
\end{align*}
Below we bound the value of this parameter for different cases on $i$, obtaining a bound
of $(K+3)\log n$ in each case, for large enough $C$.

For $-m'\leq i<0$, the smallest parameter occurs for $i=-m'$, since
$g(x)$ is monotone increasing in $[-1+\delta, 0]$.
Thus, the Poisson parameter is lower bounded by its value for $i=-m'$, which is
\begin{align*}
  &n \big [ \big( (\delta+\Delta)^K - \delta^K \big )/K - \delta^{K-1}\Delta \big ]/(K-1) \\
  \geq \,&n\Delta^K/(K(K-1)) \geq C\log n/(K(K-1)) \geq 3\log n \, ,
\end{align*}
for $C \geq 3K(K-1)$, using \eqref{eq:Delta_lb}, and $(\delta+\Delta)^K \geq \Delta^K +K\Delta \delta^{K-1}+\delta^K$.

For $0\leq i<m-m'$,  the Poisson parameter is
\begin{align*}
  n \big[ 1-\delta^{K-1}\big ] \Delta/(K-1) \geq n \Delta/(2(K-1)) \geq n\Delta^K/(K(K-1))\geq 3\log n\, ,
\end{align*}
using $\delta \geq 1/2$ and $K \geq 2$.

For $(m-m')\leq i<m$,  the Poisson parameter is
\begin{align*}
  &n (\Delta -\Delta^K((1+i)^K - i^K)/K )/(K-1) \, .
\end{align*}
A short calculation allows us to again bound this below by $(K+3)\log n$ (the bound is slack for $K>2$):
Note that
\begin{align*}
  \Delta^K((1+i)^K - i^K) \leq \Delta^K(m^K - (m-1)^K) = 1-(1-\Delta)^K\\
  \leq K\Delta - K(K-1)\Delta^2/2 + K(K-1)(K-2)\Delta^3/6\, ,
\end{align*}
where we used that $(1+i)^K - i^K$ is monotone increasing in $i$ for $i \geq 0$.
Substituting back, we obtain that the Poisson parameter is bounded by
\begin{align*}
  n (1-(K-2)\Delta/3)\Delta^2/2 \geq n\Delta^2/4
\end{align*}
for $(K-2)\Delta/3 \leq 1/2$, which occurs for sufficiently large $n$. Finally,
$\Delta^2 \geq \Delta^K$, hence $n\Delta^2/4 \geq n\Delta^K/4 \geq 3\log n$ for
$C\geq 12$.

Choosing $C=6K(K-1)$, in all cases the Poisson parameter is bounded below by $3\log n$.
It follows that
\begin{align*}
  \prob([i\Delta, (i+1)\Delta]\cap \cV^{k_1, k_2} = \emptyset ) \leq \exp(-3\log n) =1/n^{3} \, .
\end{align*}
We deduce by union bound over $i$ and De Morgan's law on \eqref{eq:Evk1k2_def} that
\begin{align}
  \prob\big(\xoverline{\Ev^{k_1,k_2}}\big) \leq 2m/n^{3} \leq n^{1/K}/n^{3} \leq 1/n^{2} \, ,
  \label{eq:Ekc_bound}
\end{align}
for large enough $n$.
Using union bound over $(k_1, k_2)$ we deduce that
\begin{align}
\prob\Big (\cup_{(k_1,k_2)}\xoverline{\Ev^{k_1,k_2}}\Big ) \leq K(K-1)/n^{2}
\label{eq:Ek1k2c_bound}
\end{align}
Combining \eqref{eq:Ekc_bound} and \eqref{eq:Ek1k2c_bound} by union bound and using De Morgan's law,
we deduce that
\begin{align*}
  \prob\Big[ \Big(\cap_k \Ev^k\Big ) \cap \Big(\cap_{(k_1,k_2)} \Ev^{k_1,k_2}\Big)\Big ] \geq 1-K^2/n^{2} 
\end{align*}
for large enough $n$. This implies that for large enough $n$, with probability at least $1-K^2/n^{2}$ we have
  \begin{align*}
    \max \left ( \max_k V^k, \max_{k_1, k_2} V^{k_1, k_2} \right ) \leq 2\Delta \leq 3 (C\log n/n)^{1/K} = O^*(1/n^{1/K}) \, ,
  \end{align*}
  implying the main result for large enough $n$ (note that $K^2/n^2 < 1/n$ for large enough $n$).
  For small values of $n$, we can simply choose $f(n,k)$ large enough to ensure that the bound holds
  with sufficient probability.
\end{proof}

\blankline.


\begin{lemma}  \label{lemma:cuboid_orderstat}
  For 
  $k \in \cK$, let $\tcR^k(\delta)$, $\tcV^k(\delta)$ and $\tV^k(\delta)$ be as defined by Eqs.~\eqref{def:cuboid}, ~\eqref{eq:cVcubdef} and \eqref{eq:Vcubdef} respectively.
  Fix $K \geq 1$. There exists $f(n) = O^*(1/n)$ such that for any $\delta \in (0,1]$, the following occurs: Let
  \begin{align}\label{def:Ev_2}
    \Ev_2 \equiv \left \{ \max_{k \in \cK} \tV^k(\delta) \leq f(n)/\delta^{K-1} \right \}\, .
  \end{align}
  Then
  $$
  \prob(\Ev_2) \geq 1-1/n \, .
  $$
\end{lemma}

\begin{proof}
  The values in the set $\tcV^k\subset [0,1]$ follow a one-dimensional
  Poisson process with rate $n\delta^{K-1}$.
  Choose $f(n) = 6\log n/n$. If $6\log n/(n\delta^{K-1}) \geq 1$ there is nothing
  to prove, since $\max_{k \in \cK} \tV^k(\delta) \leq 1$ by definition.  Hence
  assume $6\log n/(n\delta^{K-1}) < 1$.
  Divide $[0,1]$ into intervals of length $\Delta = f(n)/(3\delta^{K-1})=3\log n/(n\delta^{K-1})$
  (to simplify notation, we assume $1/\Delta \geq 2$ is an integer.
  The argument can easily be adapted to handle $n\delta^{K-1}/(3\log n)$ not an integer).
  The probability that any particular interval of length $\Delta$
  does not contain a point is no more than $\exp(-3 \log n)= 1/n^3$.
  The number of intervals of length $\Delta$ is $1/\Delta = n\delta^{K-1}/(3\log n) \leq n$
  for large enough $n$. By union bound, with probability at least $1-1/n^2$, each $\Delta$-interval
  contains at least one point, implying that $\tV^k(\delta) \leq 2 \Delta = f(n)/\delta^{K-1}$ with
  probability at least $1-1/n^2$, as required.
\end{proof}

\blankline

In this section so far we considered the rate $n$ Poisson process in $[0,1]^K$ for convenience. However, the
results we proved can easily be transported to the closely related model of $n$ points distributed i.i.d. uniformly in $[0,1]^K$.

\begin{lemma}\label{lemma:coupling}
Consider $n$ points distributed i.i.d. uniformly in $[0,1]^K$. Lemmas \ref{lemma:max_separation_hypercube} and
\ref{lemma:cuboid_orderstat} hold for this model as well.
\end{lemma}
\begin{proof}
We use a standard coupling argument along with monotonicity of the considered random variables with respect to additional points.
Let $\cP$ be a rate $n/2$ Poisson process in $[0,1]^K$. The $N$ points are distributed i.i.d. uniform $[0,1]^K$ conditioned
on the value of $N$. Let $\Ev$ be the event $N\leq n$. Clearly, $\Ev$ occurs with probability at least $1-1/n^2$.
Let $\cU$ be the process consisting of $n$ points distributed i.i.d. in $[0,1]^K$. Conditioned
on $\Ev$, we can couple the process $\cP$ with the process $\cU$  such that for every
point in the Poisson process, there is an identically located point in $\cU$.

We now show how to establish Lemma \ref{lemma:cuboid_orderstat} for process $\cU$ using such a coupling.
Note that
$\max_{k \in \cK} \tV^k(\delta)$ is monotone non-increasing as we add more points. As such, an
upper bound on this quantity continues
to hold if more points are added. For instance, consider $\max_{k \in \cK} \tV^k(\delta)$. Let
$\Ev'$ be the event that
\begin{align*}
    \max_{k \in \cK} \tV^k(\delta) \leq f(n/2)/\delta^{K-1}\,
\end{align*}
under $\cP$. The proof of Lemma \ref{lemma:cuboid_orderstat} shows that $\prob(\Ev') \geq 1-(2/n)^2$.
By union bound on $\xoverline{\Ev}$ and $\xoverline{\Ev'}$, we deduce that $\prob(\Ev\cap \Ev') \geq 1-5/n^2 \geq 1-1/n$, for large enough $n$. We deduce,
using a coupling as described above,
that with probability at least $1-1/n$,  for process $\cU$ we have
\begin{align*}
    \max_{k \in \cK} \tV^k(\delta) \leq \tilde{f}(n)/\delta^{K-1}\, ,
\end{align*}
where $\tilde{f}(n)=f(n/2)$, for large enough $n$.
(For small values of $n$, we can simply choose $\tilde{f}(n)$ large enough to ensure that the bound holds
  with sufficient probability.)
  Thus we have shown that Lemma \ref{lemma:cuboid_orderstat} holds for process $\cU$.

Lemma \ref{lemma:max_separation_hypercube} can similarly be established for process $\cU$ using that
$$\max \left ( \max_{k \in \cK} V^k, \max_{(k_1, k_2)\in \cK^{(2)}} V^{k_1, k_2}(\delta) \right )$$
is monotone non-increasing as we add more points.
\end{proof}

\blankline

We now establish another result about $n$ points $(\eps_j)_{j=1}^n$ distributed i.i.d. uniformly in $[0,1]^K$. This result is key to the proof of the tightness of Theorem~\ref{thm:general_k_q} (Proposition~\ref{prop:lower_bd}).

For $\delta \in [0,1]$ let
\begin{align}\label{def:hR_k1k2}
  \hcR^{k_1, k_2}(\delta) = \{ x \in [0,1]^K: x^{k_1} \geq x^{k_2}-\delta\,; \, x^{k_1} \geq x^{k} \,\forall k \notin \{k_1,k_2\}, k \in \cK\}
\end{align}

Let $n^{k_1, k_2}(\delta)$ be the number of points in $\hcR^{k_1, k_2}(\delta)$.

\begin{lemma}
Let $\Ev_3$ be the event that there for all $k_1, k_2 \in \cK$ we have $n^{k_1,k_2} \geq 1+n/K$.
For $\delta = \delta(n) \geq 1/n^{0.49}$, we have that $\Ev_3$ occurs with high probability.
\label{lemma:Ev3}
\end{lemma}

\begin{proof}
A short calculation shows that the volume of $\hcR^{k_1, k_2}(\delta)$
is
\begin{align}
  v=\;&\frac{1}{K-1}\left ( 1-\frac{(1-\delta)^K}{K}\right )\\
  \geq \; &\frac{1}{K}+\frac{\delta}{K-1}-\frac{\delta^2}{2}\\
  \geq \; & \frac{1+\delta}{K}
\end{align}
for $\delta \leq 2/(K(K-1))$. Now, the probability of $\eps_j \in \hcR^{k_1, k_2}(\delta)$ is exactly $v$. It follows that $n^{k_1, k_2}$ is distributed as $\textup{Binomial}(n,v)$. Notice $\E[n^{k_1, k_2}] = nv \geq n(1+\delta)/K$. We obtain
\begin{align}
\prob(n^{k_1, k_2} < 1+n/K) \leq \exp\!\big \{\!-\Omega\big(n \delta^2\big)\big\} = \exp\!\big\{\!- \Omega\big(n^{0.02}\big)\big\} = o(1)
\end{align}
using a standard Chernoff bound (e.g., see Durrett \cite{durrett2010probability}).  Using union bound over pairs $k_1, k_2$ we deduce that $\xoverline{\Ev_3}$ occurs with
probability $o(1)$, i.e., event $\Ev_3$ occurs with high probability.
\end{proof}

\setcounter{lemma}{0}
\setcounter{claim}{0}
\setcounter{remark}{0}

\section{Proof of Theorem~\ref{thm:general_k_q} upper bound}
\label{app:proofKQ}

We now present the complete proof of Theorem~\ref{thm:general_k_q}. We start by proving the lemmas stated in Section~\ref{sec:proofs}.


\begin{proof}[Proof of Lemma~\ref{lemma:type_adjacency_graph}]
Suppose not, and let $C$ be a connected component of $G(M)$ where all vertices are unmarked. Abusing notation, let $C_{\ml}$ (resp. $C_{\me}$) denote the types in $\typelab$ (resp. $\typeemp$) that are in $C$. By the definition of the marks, we know that all agents of types in $C_{\ml} \cup C_{\me}$  must be matched. Furthermore, by the definition  of $G(M)$, an agent whose type is in $C_{\ml}$ can only be matched to an agent whose type is in $C_{\me}$ and vice versa. Therefore, we must have that $\sum_{k \in C_{\ml}} n_k = \sum_{q \in C_{\me}} n_q$, which contradicts Assumption~\ref{assumpt:avoid_symmetry}.
\end{proof}

\tinyblank

\begin{proof}[Proof of Lemma~\ref{lemma:all_events_prob}]
By invoking Lemma~\ref{lemma:max_separation_hypercube}, Lemma~\ref{lemma:cuboid_orderstat} and Lemma~\ref{lemma:coupling}, for each $t$ we have that w.p. at least $1-\frac{2}{n_t}$ the event $\left( \Ev_1(t,\delta)\cap \Ev_2(t,\delta) \right)$ occurs. As the total number of types is upper bounded by $K+Q$, we apply an union bound to conclude that w.p. at least $1-\frac{2(K+Q)}{n^*}$, the event $\bigcap_{t\in  \typelab \cup \typeemp} \left( \Ev_1(t,\delta)\cap \Ev_2(t,\delta) \right)$ occurs.
\end{proof}
\tinyblank

Before moving on to the key lemmas, we introduce some definitions. Given a type $t\in \typelab \cup \typeemp$ we denote by $\nu(t)$ or simple $\nu$, the points in $t$. That is, for each agent $j$ of type $t$, we define $\nu_j$ as follows:

\[\nu_j = \left\{
  \begin{array}{l l}
   \epsilon_j & \quad \text{if $t\in \typeemp$}\\
   \eta_j & \quad \text{if $t\in \typelab$}
  \end{array} \right.\]

\blankline

For a fixed $t\in \typelab \cup \typeemp$ and $t' \in \vartheta(t)$,  let $\beta_{tt'}$ be defined as:

\[ \beta_{tt'} = \left\{
  \begin{array}{l l}
    - \alpha_{tt'} & \quad \text{if $t\in \typeemp$}\\
    \alpha_{tt'}-u(t,t') & \quad \text{if $t\in \typelab$}
  \end{array} \right.\]

Using the above notation, we can re-write the conditions in Proposition~\ref{prop:general_nec_cond} associated to a fixed type $t\in \typelab \cup \typeemp$ as follows:

\begin{enumerate}
\item[(ST)] For every $k, k' \in \vartheta(t)$:
$$ \min_{j \in t\cap M(k)}\nu^k_j - \nu^{k'}_j \geq \beta_{kt} - \beta_{k't} \geq  \max_{j \in t\cap M(k')}\nu^k_j - \nu^{k'}_j.$$
\item[(IM)]For every $k \in \vartheta(t)$:  $$\min_{j \in t\cap M(k)}\nu^k_j  \geq  \beta_{kt} \geq  \max_{j \in q \cap U}\nu^k_j.$$
\end{enumerate}

As all the $\nu$ variables are in $[0,1]$, then the above conditions can be interpreted as geometric conditions in the $[0,1]^{D(t)}$-hypercube.

The proof of Lemma~\ref{lemma:artificial_unmatched} is partitioned into two lemmas. Given a core solution $(M,\alpha)$, let the event $\EvD(t,\delta)$ be defined as:
\begin{align}\label{def:event_betas}
  \EvD(t,\delta)=\{\beta_{tz} \geq \delta \quad \forall z\in\vartheta(t)\}.
\end{align}
Lemma~\ref{lemma:unmatched_with_big_betas} below deals with $\EvD(t,\delta)$ whereas Lemma \ref{lemma:unmatched_with_small_betas} deals with
the complement $\xoverline{\EvD(t,\delta)}$. Together they imply Lemma~\ref{lemma:artificial_unmatched}.

\begin{lemma}\label{lemma:unmatched_with_big_betas}
Consider a core solution $(M,\alpha)$ and a type $t$. Let the events $\EvF_1(t)$, $\EvD(t,\delta)$ and $\Ev_2(t,\delta)$ be as defined by Eqs.~\eqref{def:event_marked}, ~\eqref{def:event_betas} and \eqref{def:event_cuboid} respectively. Under $\EvF_1(t) \cap \EvD(t,\delta) \cap \Ev_2(t,\delta)$, we  have $\max_{t'\in \vartheta(t)} \left( \alpha^{\max}_{t,t'} - \alpha^{\min}_{t,t'}\right) \leq f_2(n_t)/\delta^{D(t)-1}$, where $f_2$ is as defined in the statement of Lemma~\ref{lemma:cuboid_orderstat}.
\end{lemma}
\proof Let $D=D(t)$. Fix $k\in \vartheta(t)$ and consider the orthotope $\tcR^k=\tcR^k(t,\delta)$ as defined by Eq.~\eqref{def:cuboid}. 
As $\EvD(t,\delta)$ occurs, $\beta_{tz} \geq 1/\delta$ for all $z\in\vartheta(t)$ and therefore $\tcR^k$ can only contain points corresponding to agents in $M(k)\cup U$. 
By using the notation introduced above, condition (IM) in Proposition~\ref{prop:general_nec_cond} implies: $\alpha^{\max}_{kt} - \alpha^{\min}_{kt} \leq \min_{j \in t\cap M(k)}\nu^k_j  - \max_{j \in q \cap U}\nu^k_j$.
However,
 $\min_{j \in t\cap M(k)}\nu^k_j  - \max_{j \in q \cap U}\nu^k_j \leq \min_{j \in \tcR^k \cap M(k)}\nu^k_j  - \max_{j \in \tcR^k \cap U}\nu^k_j \leq \tV^k(t,\delta)$, where $\tV^k(t,\delta)$ is as defined by Eq.~\eqref{eq:Vcubdef}. Therefore, for each $k\in\vartheta(t)$ we must have $\alpha^{\max}_{kt} - \alpha^{\min}_{kt} \leq \tV^k(t,\delta)$.
Finally, under $\Ev_2(t,\delta)$ we have $\max_{k \in \vartheta(t)}\tV^k(t,\delta) \leq  f_2(n_t)/\delta^{D-1}$, which completes the result.
\endproof

\tinyblank

\begin{lemma}\label{lemma:unmatched_with_small_betas}
Consider a core solution $(M,\alpha)$ and a type $t$. 
Let $\EvF_1(t)$ be the event defined in Eq.~\eqref{def:event_marked}. Let the event $\Ev_1(t,\delta)$ be as defined by Eq.~\eqref{def:event_triangle}, and let the event $\xoverline{\EvD(t,\delta)}$ denote the complement of the event defined by Eq.~\eqref{def:event_betas}. Under $\EvF_1(t) \cap \xoverline{\EvD(t,\delta)} \cap \Ev_1(t,\delta)$, we  have $\max_{t'\in \vartheta(t)} \left( \alpha^{\max}_{t,t'} - \alpha^{\min}_{t,t'}\right) \leq f_1(n_t,D(t))+ \delta$, where $f_1$ is as defined in the statement of Lemma~\ref{lemma:max_separation_hypercube}.
\end{lemma}

\proof
Suppose $t\in \typeemp$. Consider the unit hypercube in $\real^K$. For each $j\in \me$ such that $\tau(j)=t$, let $\epsilon_j \in [0,1]^K$ denote the vector of realizations of $\epsilon_j^k$ for every $k\in \typelab$.
By condition (ST) in Proposition~\ref{prop:general_nec_cond}, we can partition the $[0,1]^K$ hypercube into $|\vartheta(t)|$ regions such that all the points $\epsilon$ corresponding to agents matched to $k\in \vartheta(t)$ must be contained in the corresponding region. In particular, for each $k \in \vartheta(t)$, we define $Z(k) \subseteq [0,1]^K$ to be the region corresponding to type $k$, with $Z(k) = \cap_{k'\in \vartheta(t),~k' \neq k} \{ x \in [0,1]^K: x_k - x_{k'} \geq \alpha_{k't}-\alpha_{kt} \}$. Note that the region $Z(k)$ can only contain points corresponding to agents matched to $k$ or unmatched.

Let $k^*=\argmax_{k\in \typelab}\{\alpha_{tk^*}:~ k \in \vartheta(t)\}$, and let $\cR^{k^*}=\cR^{k^*}(t)$ be as defined by Eq.~\eqref{def:R_k}. By condition (ST) in Proposition~\ref{prop:general_nec_cond}, we have that for all $k\in \vartheta(t)$:

$$\min_{j \in t\cap M(k^*)}\epsilon^{k^*}_j - \epsilon^{k}_j \geq \alpha_{kt} - \alpha_{k^*t} \geq \max_{j \in q\cap M(k)}\epsilon^{k^*}_j - \epsilon^{k}_j.$$

As $\alpha_{kt} - \alpha_{k^*t}\leq 0$ for all $k \in \vartheta(t)$, we must have $\cR^{k^*} \subseteq Z(k^*)$. Let  $V^{k^*}=V^{k^*}(t)$ be as defined in Eq.~\eqref{eq:Vkdef}. We claim that $\alpha^{\max}_{k^*,t} - \alpha^{\min}_{k^*,t} \leq V^{k^*}$. To see why this holds, consider two separate cases. First, suppose there is at least one point corresponding to an unmatched agent in $\cR^{k^*}$. By condition (IM) in Proposition~\ref{prop:general_nec_cond}, we must have $\min_{j \in t\cap M(k^*)}\epsilon^{k^*}_j  \geq  - \alpha_{k^*t} \geq  \max_{j \in t \cap U}\epsilon^{k^*}_j$. Hence, $\alpha^{\max}_{k^*,t} - \alpha^{\min}_{k^*,t} \leq  \min_{j \in t\cap M(k^*)}\epsilon^{k^*}_j - \max_{j \in t \cap U}\epsilon^{k^*}_j \leq V^{k^*}$ as desired. For the second case, suppose that all points in $\cR^{k^*}$ correspond to matched agents. As $\max_{j \in t \cap U}\epsilon^{k^*}_j \geq 0$, we must have $\alpha^{\max}_{k^*,t} - \alpha^{\min}_{k^*,t} \leq  \min_{j \in t\cap M(k^*)}  \epsilon^{k^*}_j \leq \min_{j \in \cR^{k^*} } \epsilon^{k^*}_j \leq V^{k^*}$, as the difference between $0$ and the $\min_{j \in \cR^{k^*} } \epsilon^{k^*}_j$ is upper bounded by $V^{k^*}$. Therefore, we conclude $\alpha^{\max}_{k^*,t} - \alpha^{\min}_{k^*,t} \leq V^{k^*}$.

Next, we consider the bound for any arbitrary type $k \in \vartheta(t)$. By condition (ST) in Proposition~\ref{prop:general_nec_cond}, we have that for all $k\in \vartheta(t)$:

$$\alpha^{\max}_{k^*t} + \min_{j \in t\cap M(k^*)}\epsilon^{k^*}_j - \epsilon^{k}_j \geq \alpha_{kt} \geq \max_{j \in q\cap M(k)}\epsilon^{k^*}_j - \epsilon^{k}_j + \alpha^{\min}_{k^*t}.$$

Therefore,  $$\alpha^{\max}_{kt} -  \alpha^{\min}_{kt} \leq  \alpha^{\max}_{k^*t} - \alpha^{\min}_{k^*t} + \min_{j \in t\cap M(k^*)}\epsilon^{k^*}_j - \epsilon^{k}_j - \max_{j \in q\cap M(k)}\epsilon^{k^*}_j - \epsilon^{k}_j.$$

From our previous bound, we have that $\alpha^{\max}_{k^*t} - \alpha^{\min}_{k^*t} \leq V^{k^*}$. We now want an upper bound on  $\min_{j \in t\cap M(k^*)}\epsilon^{k^*}_j - \epsilon^{k}_j - \max_{j \in q\cap M(k)}\epsilon^{k^*}_j - \epsilon^{k}_j$.
Let $\cR^{k^*,k}=\cR^{k^*,k}(t,delta)$ and $V^{k^*, k}=V^{k^*,k}(t)$  be as defined by Eqs.~\eqref{def:R_k1_k2} and \eqref{eq:Vk1k2def}. We shall show that $\min_{j \in t\cap M(k^*)}\epsilon^{k^*}_j - \epsilon^{k}_j - \max_{j \in q\cap M(k)}\epsilon^{k^*}_j - \epsilon^{k}_j \leq V^{k^*, k} + \delta$. Recall that, under $\xoverline{\EvD(t,\delta)}$, we have $\delta \geq \alpha_{k^*t}$.

To that end, note that all points in $\cR^{k^*,k}$ must correspond to agents matched to $k^*$ or matched to $k$, as the region $\cR^{k^*,k}$  cannot contain unmatched without violating condition (IM).  Furthermore, as $\cR^{k^*} \subseteq Z(k^*)$ and $\cR^{k^*} \cap \cR^{k^*,k} \neq \emptyset$, at least one point in $\cR^{k^*,k}$ corresponds to an agent matched to $k^*$.
We now consider two separate cases, depending on whether $\cR^{k^*,k}$ contains a at least one point matched to $k$. First, suppose  $\cR^{k^*,k}$ contains a at least one point matched to $k$. Then, the bound trivially applies as $$\min_{j \in t\cap M(k^*)}\epsilon^{k^*}_j - \epsilon^{k}_j - \max_{j \in t\cap M(k)}\epsilon^{k^*}_j - \epsilon^{k}_j \leq \min_{j \in \cR^{k^*,k}\cap M(k^*)}\epsilon^{k^*}_j - \epsilon^{k}_j - \max_{j \in \cR^{k^*,k}\cap M(k)}\epsilon^{k^*}_j - \epsilon^{k}_j \leq V^{k^*, k}.$$

Otherwise, $\cR^{k^*,k}$ contains only points matched to $k^*$. In that case, $$\min_{j \in t\cap M(k^*)}\epsilon^{k^*}_j - \epsilon^{k}_j - \max_{j \in q\cap M(k)}\epsilon^{k^*}_j - \epsilon^{k}_j \leq \min_{j \in \cR^{k^*,k}}\epsilon^{k^*}_j - \epsilon^{k}_j - (1+\alpha_{k^*t}) \leq  V^{k^*, k} + \delta,$$

\noindent as desired. Overall, we have shown that:
$$\max_{k\in \vartheta(t)} \left( \alpha^{\max}_{tk} - \alpha^{\min}_{tk}\right) \leq \max\left( V^{k^*}, \max_{k\in \vartheta(t)}\left(V^{k^*}+ V^{k^*, k} + \delta \right) \right).$$

Under $\Ev_1(t,\delta)$ we have $\max\left( V^{k^*}, \max_{k} V^{k^*, k} \right)\leq f_1(n_t,K)$, implying $$\max\left( V^{k^*}, \max_{k\in \vartheta(t)}\left( V^{k^*}+ V^{k^*, k} + \delta \right)\right) \leq 2f_1(n_t,K)+\delta,$$

\noindent as desired.

To conclude, we briefly discuss the changes when $t\in \typelab$.  Consider the unit hypercube in $\real^Q$. For each $j\in \ml$ such that $\tau(j)=t$, let $\eta_j \in [0,1]^Q$ denote the vector of realizations of $\eta_j^q$ for every $q\in \typeemp$. For each $q \in \vartheta(t)$, we define $Z(q) \subseteq [0,1]^Q$ to be the region corresponding to type $q$. The main difference with the case in which $t\in \typeemp$ is that we need to define the regions $Z(q)$ in terms of the $\tilde{\eta}$ instead of $\eta$. To that end, let $\beta_{kq}=\alpha_{kq}-u(k,q)$.
By the (ST) condition in Proposition~\ref{prop:general_nec_cond}, we must have:
$$\min_{i \in t\cap M(q')} \tilde{\eta}^{q'}_i-\tilde{\eta}^{q}_i \geq \alpha_{tq'} - \alpha_{tq} \geq \max_{i \in t\cap M(q)} \tilde{\eta}^{q'}_i-\tilde{\eta}^{q}_i,$$
\noindent or equivalently,
$$\min_{i \in t\cap M(q')}{\eta}^{q'}_i-{\eta}^{q}_i \geq \beta_{tq'} - \beta_{tq} \geq \max_{i \in t\cap M(q)} {\eta}^{q'}_i-{\eta}^{q}_i.$$

By using $\beta$ instead of $\alpha$, the same geometric intuition as before applies. Then, we define $Z(q) = \cap_{q'\in \vartheta(t),~q' \neq q} \{ x \in [0,1]^Q: x_q - x_{q'} \geq \beta_{qt}-\beta_{q't} \}$.  To select $q^*$, we just select the one with smallest $\beta_{qt}$. The rest of the proof remains the same.
\endproof

\blankline

\begin{proof}[Proof of Lemma~\ref{lemma:artificial_unmatched}]
Lemma~\ref{lemma:artificial_unmatched} immediately follows from Lemmas~\ref{lemma:unmatched_with_big_betas} and \ref{lemma:unmatched_with_small_betas}.
\end{proof}

\blankline

\begin{proof}[Proof of Lemma~\ref{lemma:all_matched}]
Consider a core solution $(M,\alpha)$. Let $D=D(t)$. Fix a type $t^*\in \vartheta(t)$, and let $k^*=\argmax_{k \in \vartheta(t)} \beta_{tk}$. 
We start by showing that, under $\EvF_2(t) \cap \Ev_1(t,\delta)$, we must have $\alpha^{\max}_{tk^*} - \alpha^{\min}_{tk^*}  \leq \left(\alpha^{\max}_{t,t*}- \alpha^{\min}_{t,t*}\right) + f_1(n_t,D(t))+2\delta$. If $k^*=t^*$, the claim follows trivially. Otherwise, let $\cR^{k^*,t^*}=\cR^{k^*,t^*}(t,\delta)$ and $V^{k^*, t^*}=V^{k^*,t^*}(t,\delta)$  be as defined by Eqs.~\eqref{def:R_k1_k2} and \eqref{eq:Vk1k2def}. We show that $\min_{j \in t\cap M(k^*)}\nu^{k^*}_j - \nu^{t^*}_j - \max_{j \in q\cap M(t^*)}\nu^{k^*}_j - \nu^{t^*}_j \leq V^{k^*, t^*} + \delta$.

To that end, note that all points in $\cR^{k^*,t^*}$ must correspond to agents matched to $k^*$ or matched to $t^*$, as under $\EvF_2(t)$ all agents in $t$ are matched.  Furthermore, by the definition of $k^*$, $\cR^{k^*,t^*}$  must contain a point corresponding to an agent matched to $k^*$. We now consider two separate cases, depending on whether $\cR^{k^*,t^*}$ contains at least one point corresponding to an agent matched to $t^*$. First, suppose  $\cR^{k^*,t^*}$ contains at least one point corresponding to an agent matched to $t^*$. Then, $$\min_{j \in t\cap M(k^*)}\nu^{k^*}_j - \nu^{t^*}_j - \max_{j \in t\cap M(t^*)}\nu^{k^*}_j - \nu^{t^*}_j \leq \min_{j \in \cR^{k^*,t^*}\cap M(k^*)}\nu^{k^*}_j - \nu^{t^*}_j - \max_{j \in \cR^{k^*,t^*}\cap M(k)}\nu^{k^*}_j - \nu^{t^*}_j \leq V^{k^*, t^*}.$$

Otherwise, $\cR^{k^*,t^*}$ contains only points matched to $k^*$. In that case, $$\min_{j \in t\cap M(k^*)}\nu^{k^*}_j - \nu^{t^*}_j - \max_{j \in t\cap M(t^*)}\nu^{k^*}_j - \nu^{t^*}_j \leq \min_{j \in \cR^{k^*,t^*}}\nu^{k^*}_j - \nu^{t^*}_j - 1 \leq  V^{k^*, t^*} + \delta,$$

\noindent as desired. By condition (ST) in Proposition~\ref{prop:general_nec_cond}, we must have:

$$\alpha^{\max}_{tk^*} - \alpha^{\min}_{tk^*} \leq \alpha^{\max}_{tt^*} - \alpha^{\min}_{tt^*} +  \min_{j \in t\cap M(k^*)}\nu^{k^*}_j - \nu^{t^*}_j - \max_{j \in t\cap M(t^*)}\nu^{k^*}_j - \nu^{t^*}_j  \leq \alpha^{\max}_{tt^*} - \alpha^{\min}_{tt^*} +  V^{k^*, t^*}+ \delta$$

Next, consider an arbitrary $k \in \vartheta(t)$ with $k\neq t^*, k^*$.
By condition (ST) in Proposition~\ref{prop:general_nec_cond}, we must have:

$$\alpha^{\max}_{kt} -  \alpha^{\min}_{kt} \leq  \alpha^{\max}_{k^*t} - \alpha^{\min}_{k^*t} + \min_{j \in t\cap M(k^*)}\nu^{k^*}_j - \nu^{k}_j - \max_{j \in q\cap M(k)}\nu^{k^*}_j - \nu^{k}_j.$$

Let $\cR^{k^*,k}= \cR^{k^*,k}(t,\delta)$ and $V^{k^*, k}=V^{k^*, k}(t)$  be as defined by Eqs.~\eqref{def:R_k1_k2} and \eqref{eq:Vk1k2def}. By repeating the same arguments as before, we can show that $\min_{j \in t\cap M(k^*)}\nu^{k^*}_j - \nu^{k}_j - \max_{j \in q\cap M(k)}\nu^{k^*}_j - \nu^{k}_j \leq V^{k^*,k} + 2\delta $.
Hence,

$$\alpha^{\max}_{kt} -  \alpha^{\min}_{kt} \leq  \alpha^{\max}_{k^*t} - \alpha^{\min}_{k^*t} + V^{k^*,k} + \delta \leq \alpha^{\max}_{tt^*} - \alpha^{\min}_{tt^*} +  V^{k^*, t^*}+ V^{k^*,k} + 2\delta.$$

To conclude, note that
$$\max_{k \in \vartheta(t)} \left(\alpha^{\max}_{kt} -  \alpha^{\min}_{kt} \right) \leq \left(\alpha^{\max}_{tt^*} - \alpha^{\min}_{tt^*}\right) +  2\left(\max_{k \in \vartheta(t)} V^{k^*,k}\right) + 2\delta \leq \left(\alpha^{\max}_{tt^*} - \alpha^{\min}_{tt^*}\right) +  2f_1(n_t,D) + 2\delta, $$

\noindent where the last inequality follows from the fact that $\Ev_1(t,\delta)$ occurs by hypothesis.
\end{proof}

We can now proceed to the proof of the main theorem.

\begin{proof}[Proof of Theorem~\ref{thm:general_k_q}]
Let $n^*=\min_{t \in \typelab \cup \typeemp} n_t$. Under Assumption~\ref{assump:essential_assumpts}, we have that $n^*=\Theta(n)$. Let $\delta=1/{(n^*)}^{1/\max(K,Q)}$. For each $t\in \typelab \cup \typeemp$, let the events $\Ev_1(t,\delta)$ and  $\Ev_2(t,\delta)$ be as defined by Eqs.~\eqref{def:event_triangle} and \eqref{def:event_cuboid} respectively. We start by showing that, under $\bigcap_{t\in  \typelab \cup \typeemp} \left( \Ev_1(t,\delta)\cap \Ev_2(t,\delta) \right)$, we must have $\core \leq  O^*\left(\frac{1}{\sqrt[\max(K,Q)]{n}}\right)$.

To that end, construct the type-adjacency graph $G(M)$ as defined in Section~\ref{sec:proofs}. For each vertex $v$, we denote by $d(v)$ the minimum distance between $v$ and any marked vertex (that is, $d(v)=0$ if $v$ is marked, $d(v)=1$ if $v$ is unmarked and has a marked neighbour, and so on).  By Lemma~\ref{lemma:type_adjacency_graph}, we know that w.p.1, each connected component of $G(M)$ must contain at least one marked vertex, so $d(v)$ is well-defined for all $v$. Let $C_d=\{v\in C: d(v)=d\}$, that is $C_d$ is the set of vertices that are at distance $d$ from a marked vertex. We now show the result by induction in $d$. In particular, we show that, under $\bigcap_{t\in  \typelab \cup \typeemp} \left( \Ev_1(t,\delta)\cap \Ev_2(t,\delta) \right)$, for each $t\in C_d$ we have that $\max_{k \in \vartheta(t)} \left( \alpha^{\max}_{tk} -\alpha^{\min}_{tk} \right) \leq g_d(n^*,\max(K,Q))$ for some $g_d(n^*,\max(K,Q))=O^*(\frac{1}{n^{1/\max(K,Q)}})$.

%

We start by showing that the claim holds for the base case $d=0$. For each $t\in C_0$, either all agents in $t$ are unmatched or at least one agent is matched. In the former case, we can just ignore type $t$ as it will not contribute to the size of the core. In the latter, we note that w.p.1 the event $\EvF_1(t)$ as defined in the statement of Lemma~\ref{lemma:artificial_unmatched} must hold. Therefore, we can apply  Lemma~\ref{lemma:artificial_unmatched} to obtain $\max_{t'\in \vartheta(t)} \left( \alpha^{\max}_{t,t'} - \alpha^{\min}_{t,t'}\right) \leq \max\left(f_1(n_t,D(t))+ \delta, f_2(n_t)/\delta^{D(t)-1} \right)$, where $f_1$ and $f_2$ are as defined in the statement of the lemma.
To conclude the proof of the base case, let $$g_0(n^*,\max(K,Q))=\max\left(f_1(n^*,\max(K,Q))+\delta, f_2(n^*)/\delta^{\max(K,Q)-1}\right).$$ By the definition of $f_1$, $f_2$, and $\delta$, together with Assumption~\ref{assump:essential_assumpts}, we have $g_0(n^*,\max(K,Q))=O^*(\frac{1}{n^{1/\max(K,Q)}})$. Therefore, we have shown that, for every $t\in C_0$, 
we have $$\max_{k \in \vartheta(t)} \left( \alpha^{\max}_{tk} -\alpha^{\min}_{tk} \right) \leq g_0(n^*,\max(K,Q)).$$

Now suppose the result holds for all $d'\leq d$, we want to show it holds for $d+1$. Fix $t\in C_{d+1}$. By definition of $C_{d+1}$, we have that all agents in $t$ must be matched and therefore w.p.1, the event $\EvF_2(t)$ as defined in the statement of Lemma~\ref{lemma:all_matched} occurs. Moreover, there must exist a $t^*$ such that the vertex corresponding to $t^*$ is $C_{d}$ and $t^* \in \vartheta(t)$. By induction, 
we have that  $\left( \alpha^{\max}_{tt^*} -\alpha^{\min}_{tt^*} \right) \leq g_{d}(n^*,\max(K,Q))$ for $g_{d}(n^*,\max(K,Q))=O^*(\frac{1}{n^{1/\max(K,Q)}})$. Further, by Lemma~\ref{lemma:all_matched}, we know that under $\EvF_2(t) \cap \Ev_1(t,\delta)$, we have $\max_{t'\in \vartheta(t)} \left( \alpha^{\max}_{t,t'} - \alpha^{\min}_{t,t'}\right) \leq \left(\alpha^{\max}_{t,t*}- \alpha^{\min}_{t,t*}\right) + 2f_1(n_t,D(t))+2\delta$, where $\Ev_1(t,\delta)$ as defined by Eq.~\eqref{def:ev_1} and $f$ is as defined in the statement of Lemma~\ref{lemma:max_separation_hypercube}. 
Therefore, by letting $g_{d+1}(n^*,\max(K,Q))= g_{d}(n^*,\max(K,Q))+ 2f_1(n^*,\max(K,Q))+2\delta$, we have show that with probability at least $1-\frac{d+1}{n^*}$, we have  $\max_{k \in \vartheta(t)} \left( \alpha^{\max}_{tk} -\alpha^{\min}_{tk} \right) \leq g_{d+1}(n^*,\max(K,Q))$  with $g_{d+1}(n^*,\max(K,Q))=O^*(\frac{1}{n^{1/\max(K,Q)}})$.

Next, we note that $\max_v d(v)$ is upper bounded by $K+Q$. Hence, for every $t\in \typelab \cup \typeemp$, we have  $\max_{k \in \vartheta(t)} \left( \alpha^{\max}_{tk} -\alpha^{\min}_{tk} \right) \leq g_{K+Q}(n^*,\max(K,Q))$ for $g_{K+Q}(n^*,\max(K,Q))=O^*(\frac{1}{n^{1/\max(K,Q)}})$ and therefore

$$\max_{t\in \typelab \cup \typeemp} \max_{k \in \vartheta(t)} \left( \alpha^{\max}_{tk} -\alpha^{\min}_{tk} \right) \leq g_{K+Q}(n^*,\max(K,Q)).$$

To conclude, by Lemma~\ref{lemma:all_events_prob} we have that with probability at least $1-\frac{2(K+Q)}{n^*}$, the event $\bigcap_{t\in  \typelab \cup \typeemp} \left( \Ev_1(t,\delta)\cap \Ev_2(t,\delta) \right)$ occurs. In all other cases, we just use the fact that the size of the core is upper-bounded by a constant $C<\infty$. Hence,

\begin{eqnarray*}
\E[\core] & = & \dfrac{\sum_{(k,q)\in \typelab \times \typeemp} N(k,q) \left(\alpha^{\max}_{kq} - \alpha^{\min}_{kq}\right)}{\sum_{(k,q)\in \typelab \times \typeemp} N(k,q)} \\
& \leq & (K+Q) g_{K+Q}(n^*,\max(K,Q)) + C\frac{2(K+Q)}{n^*}\\
& = & O^* \left( \frac{1}{\sqrt[\max(K,Q)]{n} }\right)
\end{eqnarray*}
\noindent implying the main result for large enough $n$ (note that $\frac{2(K+Q)}{n^*}=\Theta^*(1/n)$).
\end{proof}

\setcounter{lemma}{0}
\setcounter{claim}{0}
\setcounter{remark}{0}

\section{Theorem \ref{thm:general_k_q} lower bound: Proof of Proposition \ref{prop:lower_bd}}
\label{app:lower_bound}



\begin{proof}[Proof of Proposition \ref{prop:lower_bd}]

\begin{claim}
  For this market, all labor agents of types different from $k_*$ will be matched in the core.
\end{claim}
\begin{proof}
  We know that there is some employer $j$ who is either unmatched or matched to a labor agent $i'$
of type $k_*$. Consider any matching where a labor agent $i$ of type $k \neq k_*$ is unmatched.
Now $\Phi(i',j) = \eps_{i'}+\eta_j^{k_*} \leq 1+1 = 2$, whereas $\Phi(i,j) \geq u(k,1) = 3$, hence the weight of
such a matching can be increased by instead matching $j$ to $i$. It follows that in any maximum
weight matching, all labor agents with type different from $k_*$ are matched. Finally,
recall that every core outcome lives on a maximum weight matching, cf. Proposition~\ref{prop:general_nec_cond}
\end{proof}

Among agents $i \in k_*$, exactly one agent will be matched, specifically agent
$i_*=\arg \max_{i \in k_*} \eta_i$. Let $j_*$ be the agent matched to $i_*$ (break ties arbitrarily).
Recall that core solutions always live on a maximum weight matching, and in case of multiple
maximum weight matchings, the set of vectors $\alpha$ such that $(M, \alpha)$ is a core solution
is the same for any maximum weight matching $M$. This allows us to suppress the matching, and talk
about a vector $\alpha$ being in the core, cf. Proposition~\ref{prop:general_nec_cond}.
The (IM) condition in Proposition \ref{prop:general_nec_cond} for the pair of types $(k_*,1)$ are
\begin{align}
  \eta_{i_*} \geq \alpha_{k_*}\geq \max_{i \in k_* \backslash i_*} \eta_i \, ,
  \label{eq:alphakstar_bounds}
\end{align}
and the slack condition $\alpha_{k_*} \geq - \eps_{j_*}^{k_*}$.
The (IM) conditions for types $(k,1)$ for $k \neq k_*$ are
\begin{align}
  3+\min_{i \in k} \eta_{i_*} \geq \alpha_{k}\geq -\min_{j \in M(k)} \eps_{j}^k \, .
  \label{eq:alphak_bounds}
\end{align}

The stability conditions are
\begin{align}
  \min_{j \in M(k)} \epsilon^k_j - \epsilon^{k'}_j \geq \alpha_{k'} - \alpha_{k} \geq
  \max_{j \in M(k')} \epsilon^k_j - \epsilon^{k'}_j \, ,
\label{eq:ex_stab}
\end{align}
for all $k \neq k'$. It is easy to see that Eq.~\eqref{eq:ex_stab} with $k' = k_*$ implies
$\alpha_k \leq 2$ for all $k \neq k_*$. Hence, the upper bound in Eq.~\eqref{eq:alphak_bounds}
is slack. Consider the left stability inequality with $k'=k_*$. As Eq.\eqref{eq:alphakstar_bounds} implies $\alpha_{k_*} \geq 0$, we must have
\begin{align*}
  \alpha_k \geq - \min_{j \in M(k)} \eps_j^k - \eps_j^{k_*} \geq - \min_{j \in M(k)} \eps_j^k
\end{align*}
implying that the lower bound in \eqref{eq:alphak_bounds} is also slack. Thus a vector
$\alpha$ is in the core if and only if conditions \eqref{eq:alphakstar_bounds}
and \eqref{eq:ex_stab} are satisfied.

For simplicity, we start with the special case $K=2$, with the two types of labor
being $k$ and $k_*$.
To obtain intuition, notice that from
Eq.~\eqref{eq:alphakstar_bounds} we have
$\alpha_{k_*} \xrightarrow{\tn \rightarrow \infty} 1$ in probability,
and when we use this together with Eq.~\eqref{eq:ex_stab}
we obtain $\alpha_{k} \xrightarrow{\tn \rightarrow \infty} 2$ in probability. 
(We do not use these limits in our formal analysis below.)
Hence, we focus on Eq.~\eqref{eq:alphakstar_bounds} together with
\begin{align}
  \min_{j \neq j_*} \epsilon^k_j - \epsilon^{k_*}_j \geq \alpha_{k_*} - \alpha_{k} \geq
   \epsilon^k_{j_*} - \epsilon^{k_*}_{j_*} \, .
\label{eq:ex_stab2}
\end{align}
where $j_* = \arg \min_j \epsilon^k_j - \epsilon^{k_*}_j $. Now,
$X_j = \epsilon^k_j - \epsilon^{k_*}_j$ are distributed i.i.d. with density
$U[0,1] * U[-1, 0]$ which is
\begin{align}
  f(x)= \left \{
  \begin{array}{ll}
    1-|x| & \text{ for } |x| \leq 1\\
    0 & \text{otherwise.}
  \end{array} \right .
\end{align}
(Note that if we draw $\tn+1$ samples from this distribution, it
is not hard to see that $\E [(\min_{j \neq j_*} X_j) - X_{j_*}] = \Theta(1/\sqrt{\tn})$.)
We lower bound the expected core size as follows: Let $X_j = \epsilon^k_j - \epsilon^{k_*}_j$.
Let $\Ev$
be the event that exactly one of the $X_j$'s is in $[-1, -1+1/\sqrt{\tn}]$, and no $X_j$
is in $[-1+1/\sqrt{\tn}, -1 + 2/\sqrt{\tn}]$. Under $f$ the probability of being in
$[-1, -1+1/\sqrt{\tn}]$ is $1/(2\tn)$ and the probability of being in
$[-1+1/\sqrt{\tn}, -1 + 2/\sqrt{\tn}]$ is $3/(2\tn)$. It follows that
\begin{align}
  \prob(\Ev) = \binom{\tn+1}{1,0,\tn} \frac{1}{2n} \big ( 1-2/\tn\big )^{\tn} = \Omega(1) \, .
\label{eq:ex_Ev_prob_bd}
\end{align}

\begin{claim}\label{cl:LB_k=2}
Consider the case $K=2$. Under event $\Ev$, for any core vector $(\alpha_{k_*}, \alpha_k)$, for any
value $\alpha_k' \in [\alpha_{k_*} + 1-2/\sqrt{\tn}, \alpha_{k_*} + 1-1/\sqrt{\tn}]$, we have
that vector $(\alpha_{k_*}, \alpha_k')$ is in the core. In particular, $\core = \Omega(1/\sqrt{\tn})$.
\end{claim}
\begin{proof}
  Eq.~\eqref{eq:ex_stab2} is satisfied since event $\Ev$ holds.
Since, $\alpha_{k'}$
can take any value in an interval of length $1/\sqrt{\tn}$, it follows that
$\core = \Omega(1/\sqrt{\tn})$
under $\Ev$.
\end{proof}
Combining Claim~\ref{cl:LB_k=2} with Eq.~\eqref{eq:ex_Ev_prob_bd}, we obtain that $\E[\core] = \Omega(1/\sqrt{\tn})$ as desired.

We now construct a similar argument for $K>2$, with $\cK = \typelab \backslash \{k_*\}$ being the other labor types, all of whose agents are matched.
It again turns out that
$\alpha_{k_*} \xrightarrow{\tn \rightarrow \infty} 1$ in probability,
and when we use this together with Eq.~\eqref{eq:ex_stab}
we obtain $\alpha_{k} \xrightarrow{\tn \rightarrow \infty} 2 \,~ \forall k \in \cK$ in probability (but we do not prove or use these limits). 

Considering only the dimensions in $\cK$ (recall $|\cK|=K-1$ here) of each $\eps_j$, let $\Ev_3$ be the event as defined in Lemma \ref{lemma:Ev3} with $\delta=1/n^{0.51}$.

\begin{claim}\label{cl:alphabkuk_diff_bd}
 Let $\uk= \arg \min_{k \in \cK}\alpha_k$ and let $\bk=\arg \max_{k \in \cK}\alpha_k$. Under event $\Ev_3$, we claim that
  \begin{align}
    \alpha_{\bk}-\alpha_{\uk} \leq \delta
    \label{eq:alphabkuk_diff_bd}
  \end{align}
\end{claim}
\proof
  From Proposition \ref{prop:general_nec_cond}, we know that the set of core $\alpha$'s is a linear polytope,
  hence it is immediate to see that the set of $\theta$'s is an interval.
  Let $\uk= \arg \min_{k \in \cK}\alpha_k$ and let $\bk=\arg \max_{k \in \cK}\alpha_k$. Under
  event $\Ev_3$, we claim that 
    $\alpha_{\bk}-\alpha_{\uk} \leq \delta$. 
  We can argue this
  by contradiction: Suppose $\alpha_{\bk}-\alpha_{\uk} > \delta$. One can see that all $j$'s
  such that $\eps_j^\cK \in \hcR^{\bk,\uk}(\delta)$, cf. \eqref{def:hR_k1k2}, will be matched to type $\bk$, with the possible exception
  of $j_*$. Thus, under $\Ev_3$, the number of employers matched to type $\bk$ is bounded below by
  \begin{align*}
    n^{\bk,\uk}-1 \geq ((K-1)\tn+1)/(K-1) > \tn \, ,
  \end{align*}
  which is a contradiction, implying \eqref{eq:alphabkuk_diff_bd}.
\endproof

The above claim bounds the maximum difference between $\alpha$'s corresponding to any pair of types in $\cK$. Intuitively, note that all types in $\cK$ have the same $u$ and therefore the same distribution for the $\theta$ variables of the agents in such type. Moreover, all types in  $\cK$ have the same number of agents. Hence, one would expect the $\alpha$'s to be equal. While true in the limit, for each finite $n$ we need to account for the stochastic fluctuations in given realization. Therefore, we can show that no pair of $\alpha$'s in $\cK$ can differ by more than $\delta$. The next claim follows immediately from Claim \ref{cl:alphabkuk_diff_bd}.

\begin{claim}\label{cl:maxkeps_diff_bd}
 Let $k \in \cK$ be an arbitrary type. Under event $\Ev_3$, we claim that
  \begin{align}
    \max_{k' \in \cK} \left( \epsilon^{k'}_j - \epsilon^k_j \right) \leq \delta \quad \forall j \in M(k)
    \label{eq:alphabkuk_diff_bd2}
  \end{align}
\end{claim}
\proof
By Claim~\ref{cl:alphabkuk_diff_bd}, we have that under $\Ev_3$, $|\alpha_k - \alpha_{k'}| \leq \delta$ for all $k'\in cK$. By the stability condition in Eq.~\eqref{eq:ex_stab}, we have $$\delta \geq \alpha_k - \alpha_{k'} \geq \epsilon^{k'}_j - \epsilon^k_j \quad \forall j \in M(k),~ \forall k' \in \cK.$$ \noindent Therefore, for every $j \in M(k)$ we must have $\delta \geq \max_{k'\in\cK}\epsilon^{k'}_j - \epsilon^k_j$ as desired.
\endproof

Next, we focus on the stability conditions involving type $k_*$. For each $k\in \cK$, the stability condition is:
\begin{align}
   \epsilon^{k_*}_{j_*}-\epsilon^k_{j_*}\geq \alpha_{k} - \alpha_{k_*} \geq
   \max_{j \in M(k)} \epsilon^{k_*}_j -\epsilon^k_j, .
\label{eq:ex_stabK}
\end{align}

\noindent where $j_*$ is the employer matched to $i_*$. For each $j\in \me$, let $X_j$ be defined as $X_j = (\max_{k \in \cK}\epsilon^k_j) - \epsilon^{k_*}_j$. The $X_j$ are distributed i.i.d. with cumulative distribution $F(-1+\theta) = \theta^K/K$ for $\theta \in [0,1]$ (we will not be concerned with the cumulative for positive values). Let $\Ev$
be the event that exactly one of the $X_j$'s is in $[-1, -1+1/\tn^{1/K}]$ (this will be $X_{j_*}$), and no $X_j$
is in $[-1+1/\tn^{1/K}, -1 + 2/\tn^{1/K}]$. Under cumulative $F$, the probability of being in
$[-1, -1+1/\tn^{1/K}]$ is $1/(K\tn)$ and the probability of being in
$[-1+1/\tn^{1/K}, -1 + 2/\tn^{1/K}]$ is $2^K/(K\tn)$. It follows that
\begin{align}
  \prob(\Ev) = \binom{\tn+1}{1,0,\tn} \frac{1}{Kn} \big ( 1-2^K/(K\tn)\big )^{\tn} = \Omega(1) \, .
\label{eq:ex_Ev_prob_bd_K}
\end{align}

Clearly, under $\Ev$, we must have $j_*=\arg \min_{j \in \me} X_j$. Keeping this in mind, we state and prove our last claim.

\begin{claim}
Suppose $\Ev_3 \cap \Ev$ occurs.
Take any core vector $(\alpha_{k_*}, (\alpha_k)_{k \in \cK})$.
Then
\begin{align}
  \{\theta \in \reals: (\alpha_{k_*}, (\alpha_k+\theta)_{k \in \cK}) \textup{ is in the core}\}
\end{align}
is an interval of length at least
$1/n^{1/K}-2\delta = \Omega(1/n^{1/K})$.  In particular,  $\core \geq \Omega(1/n^{1/K})$.
\end{claim}

\begin{proof}
 Define
  \begin{align*}
    \utheta &= 1- 2/\tn^{1/K} +\delta - \alpha_\uk+\alpha_{k_*}\\
    \btheta &= 1- 1/\tn^{1/K} - \alpha_\bk+\alpha_{k_*}
  \end{align*}
  We claim that, under $\Ev_3 \cap \Ev$, we have that
  $\alpha(\theta)=(\alpha_{k_*}, (\alpha_k+\theta)_{k \in \cK}) \textup{ is in the core}$ for
  all $\theta \in [\utheta, \btheta]$. To establish this, we need to show that
  conditions \eqref{eq:alphakstar_bounds} and \eqref{eq:ex_stab} are satisfied. Since
  $\alpha$ belongs to the core, we immediately infer that \eqref{eq:alphakstar_bounds}
  holds, and also \eqref{eq:ex_stab} when $k_* \notin \{k, k'\}$ by definition of $\alpha(\theta)$.
  That leaves us  with \eqref{eq:ex_stabK}.
  Now, for any $k \in \cK$ and $\theta \in [\utheta, \btheta]$ we have
  \begin{align*}
    \alpha_k(\theta) &= \alpha_k + \theta \leq \alpha_{\bk}+ \theta \leq \alpha_{\bk}+\btheta
    = 1-1/\tn^{1/K}+\alpha_{k_*} \leq \eps_{j_*}^{k_*}-\eps_{j_*}^k +\alpha_{k_*} \, ,
  \end{align*}
  where used the definitions of $\bk$ and $\btheta$, and the fact that $\Ev$ occurs (so
  $1-1/\tn^{1/K} \leq \eps_{j_*}^{k_*}-\eps_{j_*}^k$).
  This establishes the left inequality in \eqref{eq:ex_stabK}.
  Similarly, for any $k \in \cK$ we have
  \begin{align*}
    \alpha_k(\theta)\, &= \alpha_k + \theta \geq \alpha_{\uk}+ \theta \geq \alpha_{\uk}+\utheta
    = 1-2/\tn^{1/K}+ \delta + \alpha_{k_*} \\
    &\geq \epsilon^{k_*}_j -\max_{k' \in \cK}\eps_{j}^{k'} + \delta + \alpha_{k_*} \geq   \epsilon^{k_*}_j - \epsilon^{k}_j + \alpha_{k_*} \qquad \forall j \in M(k) \, ,
  \end{align*}
  where used the definitions of $\uk$ and $\utheta$ for the first two inqualities, and the fact that $\Ev$ occurs
  (so $1-2/\tn^{1/K}\geq \epsilon^{k_*}_j -\max_{k' \in \cK}\eps_{j}^{k'},~\forall j \in M(k)$). 
  Finally, the last inequality follows from $\Ev_3$ and 
  Claim~\ref{cl:maxkeps_diff_bd} (which implies 
  $-\max_{k' \in \cK}\eps_{j}^{k'} + \delta \geq - \epsilon^{k}_j$ for $j \in M(k)$).
  This establishes the right inequality in \eqref{eq:ex_stabK}. Thus, we have
  shown that $\alpha(\theta)$ is in the core for all $\theta \in [\utheta, \btheta]$.
  The length of this interval is $1/\tn^{1/K} - (\alpha_{\bk}-\alpha_{\uk}) - \delta \geq 1/\tn^{1/K}-2\delta = \Omega(1/\tn^{1/K})$,
  using \eqref{eq:alphabkuk_diff_bd}. Therefore, that $\E[\core]=\Omega(1/\tn^{1/K})$ under $\Ev_3 \cap \Ev$.
\end{proof}
Using Lemma \ref{lemma:Ev3} and Eq.~\eqref{eq:ex_Ev_prob_bd_K} we have
\begin{align*}
  \prob(\Ev_3 \cap \Ev) = \Omega(1) \, .
\end{align*}
 Combining with the claim above we obtain that $\E[\core] = \Omega(1/n^{1/K})$.
\end{proof}

\setcounter{lemma}{0}
\setcounter{claim}{0}
\setcounter{remark}{0}

\section{Proof of Theorem~\ref{thm:K-1-emp}}
\label{app:K-1}

We start by restating Theorem~\ref{thm:K-1-emp} and discussing the structure of the proof.

\begin{theorem*}[Restatement of Theorem~\ref{thm:K-1-emp}]
Consider the setting in which $K\geq2$, $Q=1$, $n_{\me} > n_{\ml}$ and let $m = n_{\me} - n_{\ml}$. In addition, suppose that $u(k,1)\geq 0$ for all $k\in \typelab$. Then, under Assumption~\ref{assump:essential_assumpts}, we have $\E[\core] \leq O^*\left(\frac{1}{n^{\frac{1}{K}}m^{\frac{K-1}{K}}}\right)$.
\end{theorem*}

Note that Assumption \ref{assumpt:avoid_symmetry} is automatically satisfied under the hypotheses of the theorem.

The idea of the proof is as follows. First, we show a bound on the expectation of $\min_{k\in \typelab}\{ \alpha^{\max}_k - \alpha^{\min}_k\}$. In particular, we show that $\E\left[ \min_{k\in \typelab} \{ \alpha^{\max}_k - \alpha^{\min}_k\} \right] = O^*\left(\frac{1}{n^{\frac{1}{K}}m^{\frac{K-1}{K}}}\right)$. To do so, we note that by condition (IM) in Proposition~\ref{prop:general_nec_cond},  we must have 
$$\min_k \left(\alpha^{\max}_k - \alpha^{\min}_k \right) 
\leq  \min_{k\in \typelab} \left(\min_{j \in M(k)} \epsilon^k_j -\max_{j \in U} \epsilon^k_j \right).$$ 
\noindent Then, we consider two separate cases to prove the result, depending the size of the imbalance. When $m\leq \log(n)$, the result is shown in Lemma~\ref{lem:bd_of_diff_alpha_small_m_gen_k}, which we prove via an upper bound on  $\min_{k\in \typelab} \left(\min_{j \in M(k)} \epsilon^k_j\right)$. On the contrary, when $m\geq \log(n)$, the result is shown in Lemma~\ref{lem:bd_min_diff_alpha_gen_k}. The proof of Lemma~\ref{lem:bd_min_diff_alpha_gen_k} relies mainly on the geometry of a core solution which (roughly) allows us to first control the largest of the $\alpha$'s (all $\alpha$'s must be negative i in the core since some employers are unmatched, and we control, roughly, the least negative $\alpha$).

Next, we then show that, for every pair of types $k,q \in \typelab$ we must have
$$
\E\left[\min_{j \in M(k)} (\epsilon^k_j - \epsilon^q_j) - \max_{j \in M(q)} (\epsilon^k_j - \epsilon^q_j)\right] = O^*\left(\frac{1}{n}\right)\, .
$$
By Condition (ST) in Proposition~\ref{prop:general_nec_cond}, this implies that  for fixed $k,q\in \typelab$, the expected maximum variation in $\alpha_k - \alpha_q$ in the core is bounded by $O^*\left(\frac{1}{n}\right)$.

Finally, we use the bounds in the first two steps to argue that, for every type $k\in\mathcal{T}_{\ml} $, $$\E\left[\alpha^{\max}_k - \alpha^{\min}_k\}\right] =  O^*\left( \frac{1}{n^{\frac{1}{K}}m^{\frac{K-1}{K}}} \right),$$ which implies $\E[\core] = O^*\left( \frac{1}{n^{\frac{1}{K}}m^{\frac{K-1}{K}}} \right)$. This is done in the proof of Theorem~\ref{thm:K-1-emp}.

\blankline

We now show our bound on $\E\left[ \min_k\{ \alpha^{\max}_k - \alpha^{\min}_k\} \right]$. To that end, let $Z_k = \min_{j \in M(k)} \epsilon^k_j$ and $U_k=  \max_{j \in U} \epsilon^k_j$. By Condition (IM) in Proposition~\ref{prop:general_nec_cond}, $\E\left[ \min_k |\alpha^{\max}_k - \alpha^{\min}_k|\right] \leq  \E[\min_k\{Z_k- U_k\}]$, and therefore we will focus on bounding $\E[\min_k\{Z_k- U_k\}]$. As a reminder, we have defined $m= n_{\me}-n_{\ml}$ and  $\delta_n = \frac{\log(n)}{n^{\frac{1}{K}}m^{\frac{K-1}{K}}}$. Also, in all lemmas we are working under the assumptions of the theorem, that is, $K\geq 2, Q=1$, $n_{\me}>n_{\ml}$ and Assumption~\ref{assump:essential_assumpts}.

\blankline

\begin{lemma}\label{lem:bd_of_diff_alpha_small_m_gen_k} Suppose $m \leq 6K\log(n_{\me})$. Then, there exists a constant $C_3=C_3(K)<\infty$ such that $ \E\left[\min_k\{\alpha^{\max}_k- \alpha^{\min}_k\}\right]\leq 2C_3\frac{\log(n)}{n^{\frac{1}{K}}m^{\frac{K-1}{K}}}$.
\end{lemma}

\proof
Let $Z_k = \min_{j \in M(k)} \epsilon^k_j$, $U_k=  \max_{j \in U} \epsilon^k_j$ and $\delta_n=\frac{\log(n)}{n^{\frac{1}{K}}m^{\frac{K-1}{K}}}$.
By Condition (IM) in Proposition~\ref{prop:general_nec_cond}, $\E[\min_k\{\alpha^{\max}_k- \alpha^{\min}_k\}] \leq \E[\min_k\{Z_k- U_k\}]$. As $U_k$ is a non-negative  random variable, we have $\E[\min_k\{Z_k- U_k\}] \leq \E[\min_k\{Z_k\}]$.  Therefore,
$$\E[\min_k\{\alpha^{\max}_k- \alpha^{\min}_k\}] \leq  \E\left[\min_k\{Z_k- U_k\}\right] \leq   \E\left[\min_k\{Z_k\}\right] \leq  C_3\delta_n  + \Pr\left(\min_{k}Z_k \geq C_3\delta_n\right),$$
using $Z_k \leq 1$.

To finish the proof, it suffices to show that $\Pr\left(\min_{k}Z_k \geq C_3\delta_n\right) \leq C_3\delta_n$.
Hence, our next step is to bound $\Pr\left(\min_{k}Z_k \geq C_3\delta_n\right)$. Now $\min_{k}Z_k \geq C_3\delta_n$ implies that all $j$ such that $\eps_j \in \left[0, C_3\delta_n\right]^K$ are unmatched. But there are only $m$ unmatched employers. It follows that
\begin{eqnarray*}
\Pr\left(\min_{k}Z_k \geq C_3\delta_n\right)& \leq & \Pr\left(\textrm{at most $m$ points in the hypercube } \left[0, C_3\delta_n\right]^K\right)\\
\end{eqnarray*}

Let $X\sim\textrm{Bin}\left(n_{\me}, \left( C_3\delta_n\right)^K\right)$ be defined as the number of points, out of $n_{\me}$ in total, that fall in the hypercube $\left[0, C_3\delta_n \right]^K$. By assumption, $m \leq 6K\log(n) \Rightarrow (C_3\delta_n)^K \geq (C_3\log n/m)^K/n \geq 2^K \log n^K/n \geq 4(\log n)^2/n$ defining $C_3 \geq 12K$ and using $K\geq 2$. Further using $n \leq 2n_{\me}$ we
obtain $\E[X] = n_{\me} \left( C_3\delta_n\right)^K \geq (n/2) 4(\log n)^2/n = 2(\log n)^2$. It follows that
$$\Pr\left(\min_{k}Z_k  \geq C_3 \delta_n\right) \leq  \Pr\left(X \leq 6K\log(n)\right) \leq \exp(-\Omega((\log n)^2))\leq \frac{1}{n} \leq C_3\delta_n $$

\noindent where the second inequality was obtained by applying the Chernoff bound. Hence, we have shown that

$$\E[\min_k\{\alpha^{\max}_k- \alpha^{\min}_k\}] \leq   \E\left[\min_k\{Z_k\}\right] \leq  C_3\delta_n  + \Pr\left(\min_{k}Z_k \geq C_3\delta_n\right) \leq 2C_3\delta_n,$$

which completes the proof.
\endproof

\blankline

We now establish an upper bound for the case in which $m \geq 6K\log(n_{\me})$. For the following results up to Lemma~\ref{lem:bd_min_diff_alpha_gen_k} we shall assume $m \geq 6K\log(n_{\me})$.

Before we move on, we briefly give some geometric intuition regarding the problem. For each agent $j \in \me$, let $\epsilon_j=(\epsilon^1_j, \ldots, \epsilon^K_j)$ denote the profile of values assigned by the $K$ types of agents in $\ml$ to agent $j$. Given our stochastic assumptions, all points $\epsilon_j$ will be distributed in the $[0,1]^K$ hypercube. Using Proposition~\ref{prop:general_nec_cond}, we can partition the $[0,1]^K$-hypercube into $K+1$ disjoint regions: $K$ of them containing the $n_k$ points corresponding to agents matched to type $k$ ($1\leq k \leq K$), and one region containing all unmatched agents. Furthermore, the region containing the unmatched agents is an orthotope\footnote{An orthotope (also called a hyperrectangle or a box) is the generalization of a rectangle for higher dimensions} that has the origin as a vertex. This follows for the (IM) constraints in  Proposition~\ref{prop:general_nec_cond}.

To that end, let $\mathcal{O}$ be the set of $K$-orthotopes contained in $[0,1]^K$ that have the origin as a vertex.
Suppose $R$ is expanded by the same amount $\theta$ in each coordinate direction. Define $D(R)$ as the smallest value of $\theta$ such that an additional point $\eps_j$ is contained in the expanded orthotope. (If one of the side lengths becomes 1 before an additional point is reached, then define $D(R)=0$. This will never occur for $R$ that contains only the unmatched agents.)
As usual, let $Z_k = \min_{j \in M(k)} \epsilon^k_j$ and $U_k=  \max_{j \in U} \epsilon^k_j$. We want to show that  $\E\left[\min_k\{Z_k- U_k\}\right] \leq C_5\delta_n$, for some constant $C_5=C_5(K)<\infty$. To that end, note that $\min_k\{Z_k-U_k\}$ is equal to $D(R)$ for some orthotope $R\in \mathcal{O}$. In particular, $\min_k\{Z_k-U_k\}$ is equal to $D(R)$ when $R$ is the orthotope that ``tightly" contains all the $m$ points in $U$.


\blankline
For $R \in \mathcal{O}$, let $V(R)$ be defined as the volume of $R$. In addition, we define $|R|$ to be the number of points contained in $R$. We start by showing that, given that $m\geq 6K\log(n)$, an orthotope in $\mathcal{O}$  of volume less than $\frac{m}{4n_{\me}}$ in extremely unlikely to contain $m$ points.

\begin{lemma}\label{lem:bounds_on_size_of_orthotope}
Suppose $m\geq 6K\log(n)$. For $R \in \mathcal{O}$ such that $V(R) < \frac{m}{4n_{\me}}$, we have $\Pr\left(|R|= m\right) \leq  \frac{1}{n^{K+1}}$, where $V(R)$ denotes the volume and $|R|$ denotes the number of points in $R$.
\end{lemma}
\proof Let $X$ denote number of points in an orthotope in $\mathcal{O}$ of volume $\frac{m}{4n_{\me}}$. Then, $X\sim \textrm{Bin}\left(n_{\me}, \frac{m}{4n_{\me}} \right)$. We have $\mu=\E[X]=m/4$. Using a  Chernoff bound we have,
\begin{align*}
 \Pr(X\geq m) = \Pr (X \geq 4 \mu) \leq (e^3/4^4)^{m/4} \leq \exp(-m/4)
\end{align*}
Now $m/4 \geq 6K\log n/4 \geq (K+1)\log n$, using $K\geq 2 $. Substituting back
we obtain $\Pr(X\geq m)\leq \exp(-(K+1) \log n) = 1/n^{K+1}$. But $|X|$ stochastically dominates
$|R|$ since $V(R)< \frac{m}{4n_{\me}}$. The result follows.
\endproof

\blankline

Our next step will be to bound $\Pr\left(D(R)>C_4\delta\Mid E \right),$ for $R \in \mathcal{O}$ and some constant $C_4=C_4(K)<\infty$ where $E$ is the event defined as $E=\{ |R| =m, ~V(R) \geq \frac{m}{4n_{\me}}\}$. 

\blankline


\begin{lemma}\label{lem:bound_on_D(R)_gen_k}
There exists some constant $C_4=C_4(K)<\infty$ such that, for all $R \in \mathcal{O}$ with $V(R) \geq \frac{m}{4n_{\me}}$, we have that $P\left(D(R)>C_4\delta_n\Mid |R| =m\right) \leq \frac{1}{n^{K+1}}$, where $\delta_n = \frac{\log(n)}{n^{\frac{1}{K}}m^{\frac{K-1}{K}}}$.
\end{lemma}
\proof
Conditioned on $|R|=m$, the remaining $n_{\ml}=n_{\me}-m$ points are distributed uniformly i.i.d. in the complementary region of volume $(1-V(R))$.

Let $F_{C_4\delta_n}$ denote the region swept when $R$ is expanded by $C_4\delta_n$
along each coordinate axis. Clearly, $D(R)>C_4\delta_n$ if and only if region $F_{C_4\delta_n}$  contains no points.

Let $X$ denote the number of points in $F_{C_4\delta_n}$, and let $p$ denote the volume of $F_{C_4\delta_n}$. Then, $X\sim \textrm{Bin}(n_{\ml},p/(1-V(R)))$ and hence stochastically dominates $\textrm{Bin}(n_{\ml},p))$. Note that such a volume $p$ is at least the volume obtained when expanding the hypercube of side $\ell=\sqrt[K]{\frac{m}{4n_{\me}}}$ by $C_4\delta_n$ along each direction and therefore, $p \geq K \ell^{(K-1)} C_4\delta_n$. Hence,
\begin{align*}
&P(D(R)>C_4\delta_n) = \Pr(X=0) \leq (1-p)^{\ml} \leq \exp \left \{ -\Omega(np)\right \}\\
\leq\; &\exp \{ -\Omega(n(\frac{m}{n})^{(K-1)/K}C_4\delta_n)\} = \exp \{ -\Omega(C_4 \log n)\} \leq
\frac{1}{n^{K+1}}\, ,
\end{align*}
for appropriate $C_4$, where we have used Assumption~\ref{assump:essential_assumpts}.


\endproof

\blankline

\begin{lemma}\label{lem:bd_min_diff_alpha_gen_k}
Suppose $m \geq 6K\log(n)$. Then, there exists a constant $C_5=C_5(K)<\infty$, such that $\E\left[\min_k\{ \alpha^{\max}_k - \alpha^{\min}_k\}\right] \leq C_5\frac{\log(n)}{n^{\frac{1}{K}}m^{\frac{K-1}{K}}}$.
\end{lemma}
\proof
Let $Z_k = \min_{j \in M(k)} \epsilon^k_j$, $U_k=  \max_{j \in U} \epsilon^k_j$ and $\delta_n=\frac{\log(n)}{n^{\frac{1}{K}}m^{\frac{K-1}{K}}}$. By Condition (IM) in Proposition~\ref{prop:general_nec_cond}, we know that $\alpha^{\max}_k - \alpha^{\min}_k \leq  Z_k-U_k$. Then, $$\E\left[\min_k\{ \alpha^{\max}_k - \alpha^{\min}_k\}\right]~ \leq \E\left[\min_k\{Z_k-U_k\} \right].$$
In addition, $\min_k\{Z_k-U_k\}$ is equal to $D(R)$ for some orthotope $R\in \mathcal{O}$. In particular, $\min_k\{Z_k-U_k\}$ is equal to $D(R)$ when $R$ is the orthotope that ``tightly" contains all the $m$ points in $U$. Define $\mathcal{R}=\{R \in \mathcal{O}:~|R|=m\}$. Then, $$\E\left[\min_k\{Z_k-U_k\} \right] ~\leq ~ \E\left[\max_{R \in \mathcal{R}} \left\lbrace\textrm{D(R)}\right\rbrace\right].$$

To bound $\E\left[\max_{R \in \mathcal{R}} \left\lbrace\textrm{D(R)}\right\rbrace\right]$, consider the grid that results from dividing each of the $K$ coordinate axes in the hypercube into intervals of length $1/n$. Let $\Delta$ denote that grid. Suppose we just consider orthotopes in the grid, that is, the orthotopes whose sides are multiples of $\frac{1}{n}$. Let $\mathcal{R}_\Delta = \{R\in \mathcal{R}:~ R \in \Delta\}$. Then,

$$\max_{R \in \mathcal{R}} \left\lbrace\textrm{D(R)}\right\rbrace \leq \max_{R \in \mathcal{R}_\Delta} \left\lbrace\textrm{D(R)}\right\rbrace + \frac{1}{n},$$

\noindent and,
$$\E\left[\max_{R \in \mathcal{R}} \left\lbrace\textrm{D(R)}\right\rbrace\right] \leq \E\left[\max_{R \in \mathcal{R}_\Delta} \left\lbrace\textrm{D(R)}\right\rbrace\right] + \frac{1}{n}.$$

Hence, we just need a bound for $\E\left[\max_{R \in \mathcal{R}_\Delta} \left\lbrace\textrm{D(R)}\right\rbrace\right]$. Let $V_*= \frac{m}{4n}$. Note that $\textrm{D(R)} \leq 1$ for all $R\in \mathcal{O}$ and therefore,
\begin{eqnarray*}\E\left[\max_{R \in \mathcal{R}_\Delta} \left\lbrace\textrm{D(R)}\right\rbrace\right] & \leq & \E\left[\max_{R \in \mathcal{R'}_\Delta} \left\lbrace\textrm{D(R)}\right\rbrace\right] +  \Pr\left(\min_{R \in \mathcal{R}_\Delta}V(R) < V_*\right)
\end{eqnarray*}
where $\mathcal{R'}_\Delta = \{R \in \mathcal{R}_\Delta: V(R)\geq V_* \}$.
Now, by union bound
\begin{align*}
  \Pr\left(\min_{R \in \mathcal{R}_\Delta}V(R) < V_*\right) \leq \sum_{R \in \Delta: V(R)< V_*} \Pr (|R|=m) \leq n^K \cdot 1/n^{K+1} = 1/n \, .
\end{align*}
using $|\{R \in \Delta: V(R)< V_*\}| \leq |\{R \in \Delta\}| = n^{K}$ and Lemma~\ref{lem:bounds_on_size_of_orthotope}.


Further,
\begin{align*}
  \E\left[\max_{R \in \mathcal{R'}_\Delta} \left\lbrace\textrm{D(R)}\right\rbrace\right]
 \leq  \E\left[\max_{R \in \Delta: V(R) \geq V_*} \left\lbrace\textrm{D(R)} \ind(|R|=m)\right\rbrace\right]
\end{align*}
Now, %
\begin{align*}
  &\Pr\left [\max_{R \in \Delta: V(R) \geq V_*} \left\lbrace\textrm{D(R)} \ind(|R|=m) \right \rbrace >C_4\delta_n \right  ]\\
\leq \; &\sum_{R \in \Delta: V(R) \geq V_*} \Pr(|R|=m) \Pr [\textrm{D(R)}>C_4 \delta_n| |R|=m]\\
\leq \; &\sum_{R \in \Delta: V(R) \geq V_*} 1 \, \cdot \, 1/n^{K+1} \leq n^K/n^{K+1} = 1/n
\end{align*}
using a union bound and Lemma~\ref{lem:bound_on_D(R)_gen_k} to bound the probability of $D(R) \geq C_4\delta_n$.
It follows that
\begin{align*}
  \E\left[\max_{R \in \mathcal{R'}_\Delta} \left\lbrace\textrm{D(R)}\right\rbrace\right]
& \leq &  1 \, \cdot \, \Pr\left [\max_{R \in \Delta: V(R) \geq V_*} \left\lbrace\textrm{D(R)} \ind(|R|=m) \right \rbrace  >C_4\delta_n \right ] + C_4 \delta_n
= 1/n + C_4 \delta_n
\end{align*}

Substituting the individual bounds back, we obtain
$$
\E\left[\max_{R \in \mathcal{R}} \left\lbrace\textrm{D(R)}\right\rbrace\right] =  C_4 \delta_n + 2/n \leq C_5 \delta_n\, .
$$
 defining $C_5 = C_4+2$ and using $1/n\leq \delta_n$.

Overall,
$$\E\left[\min_k\{ \alpha^{\max}_k - \alpha^{\min}_k\}\right]~ \leq~ \E\left[\min_k\{Z_k-U_k\} \right] ~\leq ~ \E\left[\max_{R \in \mathcal{R}} \left\lbrace\textrm{D(R)}\right\rbrace\right]\leq C_5 \delta_n\,$$ as claimed.
\endproof

\blankline

We now proceed to show that, for every pair of types $k,q \in \mathcal{T}_{\ml}$ we have $$\E\left[\min_{j \in M(k)} (\epsilon^k_j - \epsilon^q_j) - \max_{j \in M(q)} (\epsilon^k_j - \epsilon^q_j)\right] \leq C_2\frac{\log(n_{\me})}{n_{\me}}.$$
\noindent for appropriate $C_2=C_2(K)< \infty$. 
This result is shown in Lemma~\ref{lem:bound_on_the_differences_gen_k}. Along the way, we establish a couple of intermediate results. 
\blankline

Let $Z_k = \min_{j \in M(k)} \epsilon^k_j$ and $U_k = \max_{j \in U} \epsilon^k_j$. Note that $Z_k$ is an upper bound for $-\alpha_k$. By the definition of $Z_k$, all the points corresponding agents in $M(k)$ must be contained in the orthotope $[1-Z_k,1]\times [0,1]^{K-1}$. The following proposition establishes that $Z_k$ cannot be arbitrarily close to $1$.

\blankline

\begin{lemma}\label{lem:no_unmatched_agent_below_delta_gen_k} Given a constant $c \in \real$, let the event $E_{c}$ be defined as $E_{c}=\{\max_k\min_{j \in M(k)} \epsilon^k_j \leq 1-c\}$. Then, there exist constants $\theta=\theta(K)>0$ and $C_6=C_6(K)>0$ such that, for large enough $n$, $E_{\theta}$ occurs with probability at least $1-\exp\left(- C_6 n\right)$.
\end{lemma}

\proof
Let $Z_k = \min_{j \in M(k)} \epsilon^k_j$. The proof follows from the previous observation that all the points corresponding agents in $M(k)$ must be contained in the orthotope of volume $(1-Z_k)$. Let $C<\infty$ be such that $\dfrac{n_{\me}}{n_{\ml}} \leq C$. By Assumption~\ref{assump:essential_assumpts}, such a $C$ must exist. Furthermore, by Assumption~\ref{assump:essential_assumpts}, there must exists $C_K \in \reals$ such that $n_k \geq C_K n$ for all $k \in \typelab$.  Let $n_{\me}$ be the total number of points in the cube $[0,1]^K$. Let $X$ denote the number of points out of the $n_{\me}$ ones that fall in the rectangle defined by $[1-\theta, 1][0,1]^{K-1}$. Then, $X\sim \textrm{Bin}(n_{\me}, \theta)$. Suppose we set $\theta< \frac{C_K}{2C}$.
Then, for large enough $n$ and appropriate $C_6 > 0$ we have

$$\Pr(Z_k > 1-\theta) \leq \Pr(X \geq C_Kn_{\ml}) \leq \Pr\left(X \geq \frac{C_Kn_{\me}}{C}\right) \leq \exp\left(- 2C_6 n\right) \leq (1/K) \exp\left(- C_6 n\right)  $$

\noindent where we have used a Chernoff bound, $2n_{\me}\geq n$, and $\exp(-C_6 n) \leq (1/K)$ for large enough $n$. The result follows from a union bound over possible $k$.
\endproof

\blankline

\begin{remark} \label{rem:Gkq}
Let $\theta$, $E_{\theta}$ and $C_6$ be as defined in the statement of Lemma~\ref{lem:no_unmatched_agent_below_delta_gen_k}. Define $G_{k,q}$  as  $$G_{k,q}=\left\lbrace x \in [0,1]^K:~ \left(x_ k \geq 1-\frac{\theta}{2} \text{ or } x_ q\geq 1-\frac{\theta}{2}\right) \textrm{ and } x_r<\frac{\theta}{2}\textrm{ for all } 1 \leq r \leq K,~ r\neq k,q\right\rbrace.$$ Under event $E_\theta$, we must have $G_{k,q} \subseteq M(k)\cup M(q)$.
\end{remark}

The above remark follows from Lemma~\ref{lem:no_unmatched_agent_below_delta_gen_k} and the definition of $G_{k,q}$. If $j: \eps_j  \in G_{k,q}$ were matched to a type $k' \notin \{k,q\}$, that will contradict maximality of the matching as, by swapping the matches of $j': j' \in M(k), \eps_{j'}^k=Z_k$ and $j$, the overall weight of the matching strictly increases. A similar argument rules out $j$ being unmatched.

%
%

\begin{lemma}\label{lem:max_sep_in_Gkq}
Let $G_{k,q}$ be as in Remark~\ref{rem:Gkq}, and let $\theta$ be as defined in Lemma~\ref{lem:no_unmatched_agent_below_delta_gen_k}. Define $G'_{k,q}$ as follows:
$$G'_{k,q} = G_{k,q} \cap  \{x \in [0,1]^K, |x_k-x_q| \leq 1-\theta\}$$

Let $\mathcal{V}^{kq} = \{x:~ x=\epsilon^k_j - \epsilon^q_j,~ \epsilon_j\in G'_{k,q}\}$, and let $$V^{kq}=\max(\text{Difference between consecutive values in } \cV^{kq} \cup \{-1+\theta,1-\theta\}).$$ Then, there exists a function $f(n)=O^*(1/n)$ such that $\Pr\left(\xoverline{\Ev^{kq}}  \right) \leq 1/n$ where $\Ev^{kq}$ is the event that ${V}^{kq} \leq f(n)$.
\end{lemma}

The proof of Lemma~\ref{lem:max_sep_in_Gkq} is omitted as the required analysis is similar to (and much simpler than) that leading to Lemma~\ref{lemma:max_separation_hypercube}. Essentially, $V^{kq}$ consists of values taken by $\Theta(n)$ points distributed uniformly and independently in $[-1+\theta, 1-\theta]$, so, with high probability, no two consecutive values are separated by more than $f(n)=O(\log n / n)$.

\blankline

In the next lemma we bound the difference between every pair of $\alpha$'s.

\begin{lemma}\label{lem:bound_on_the_differences_gen_k} 
Consider types $k,q\in \typelab$ and let $f$ be as defined in the statement of Lemma~\ref{lem:max_sep_in_Gkq}.
Under event $E_{\theta}\cap \Ev^{kq}$, in every stable solution we must have that $(\alpha_q^{\max} - \alpha_q^{\min}) \leq 2f(n) + (\alpha_k^{\max} - \alpha_k^{\min})$.
\end{lemma}

\begin{proof}
We claim that under $E_\theta$, we must have $\alpha_q- \alpha_k$ varies within a range of no more than $V^{kq}$ within the core,
where $V^{kq}$ is as defined in the statement of Lemma~\ref{lem:max_sep_in_Gkq}.
By Remark~\ref{rem:Gkq}, under event $E_{\theta}$ we must have $G'_{kq}\subset M(k)\cup M(q)$, where $G'_{kq}$ is as defined in the statement of Lemma~\ref{lem:max_sep_in_Gkq}. Suppose that $G'_{kq}$ contains at least one vertex matched to type $k$ and one to type $q$. Then, by Condition~(ST) in Proposition~\ref{prop:general_nec_cond} we must have:
\begin{eqnarray*}
(\alpha_q - \alpha_k)^{\text{max}} - (\alpha_q - \alpha_k)^{\text{min}}  & \leq & \min_{j \in M(k)}\{ \epsilon^{k}_j - \epsilon^q_j \} - \max_{j \in M(q)}\{ \epsilon^{k}_j - \epsilon^q_ j\} \\
& \leq & \min_{j \in M(k)\cap G'_{k,q}} \{ \epsilon^{k}_j - \epsilon^q_j \} - \max_{j \in M(q) \cap G'_{k,q}}\{ \epsilon^{k}_j - \epsilon^q_j\} \\
& \leq & V^{kq}
\end{eqnarray*}
Next, consider the case in which all vertices in $G'_{kq}$ are matched to type $k$ (the analogous argument follows if they are all matched to type $q$). Under event $E_{\theta}$, by Condition~(IM) in Proposition~\ref{prop:general_nec_cond} we must have $0 \leq - \alpha_k \leq 1-\theta$ and $0 \leq -\alpha_q \leq 1-\theta$. Therefore, $\alpha_q-\alpha_k \in [-1+\theta, 1-\theta]$. In addition, by Condition~(ST) in Proposition~\ref{prop:general_nec_cond} we must have $\alpha_q - \alpha_k \leq \min_{j \in M(k)}\{ \epsilon^{k}_j - \epsilon^q_j \}$. However,
\begin{eqnarray*}
(\alpha_q - \alpha_k)^{\text{max}} - (\alpha_q - \alpha_k)^{\text{min}}& \leq & \min_{j \in M(k)}\{ \epsilon^{k}_j - \epsilon^q_j \} -(- 1 + \theta)\\
& \leq & \min_{j \in M(k)\cap G'_{k,q}}\{ \epsilon^{k}_j - \epsilon^q_j\} - (-1 + \theta) \\
& = & \min_{j \in G'_{k,q}}\{ \epsilon^{k}_j - \epsilon^q_j\} - (- 1 + \theta) \\
& \leq & V^{kq}
\end{eqnarray*}

 It follows that $(\alpha_q^{\max} - \alpha_q^{\min}) \leq 2 V^{kq} + (\alpha_k^{\max} - \alpha_k^{\min})$. By definition, under $\Ev^{kq}$ we have $V^{kq} \leq f(n)$, which completes the proof.
\end{proof}

\blankline

Finally, we complete the last step of the proof by showing the main theorem.

\proof[Proof of Theorem \ref{thm:K-1-emp}]
By definition,  $\core = \sum_{k=1}^K \dfrac{N(k)|\alpha^{\max}_{k} - \alpha^{\min}_{k}|}{n_{\ml}}$, where $N(k)$ is defined to be the number of agents of type $k$ that are matched. For a given instance, let $k^*= \textrm{argmin}_k{\{\alpha^{\max}_k - \alpha^{\min}_k\}}$.
Let $\Ev = E_\theta \cap (\cap_{k,q}\Ev^{k,q})$.
Note that using Lemmas \ref{rem:Gkq} and \ref{lem:max_sep_in_Gkq} and a union bound, we obtain that $$
\Pr ( \overline{\Ev}  ) \leq \Pr(\overline{E_\theta})+\sum_{k,q \in \cK: k\neq q} \Pr(\overline{\Ev^{kq}}) = O(1/n)\, .
$$
By Lemma~\ref{lem:bound_on_the_differences_gen_k}, under $\Ev$, for every $k \in \typelab$ we have $$\alpha_k^{\max} - \alpha_k^{\min} \leq 2f(n)+ \alpha^{\max}_{k^*} - \alpha^{\min}_{k^*}.$$
Therefore,
\begin{eqnarray*}
\E[\core] & \leq & \E\left[ \alpha^{\max}_{k^*} - \alpha^{\min}_{k^*} \right] + 2f(n) + \Pr ( \xoverline{\Ev}  ) \cdot  O(1)\\
        & \leq &  O\left ( \frac{\log(n)}{n^{\frac{1}{K}}m^{\frac{K-1}{K}}}\right ) + O^*(1/n) + O(1/n) \\
        & = & O^*\left(\frac{1}{n^{\frac{1}{K}}m^{\frac{K-1}{K}}} \right)
\end{eqnarray*}
\noindent where the first inequality follows from the above together with using the upperbound of $O(1)$ for the core size; the second inequality is obtained by  using the bound on $ \E\left[ \alpha^{\max}_{k^*} - \alpha^{\min}_{k^*} \right]$ from Lemma~\ref{lem:bd_of_diff_alpha_small_m_gen_k} for $m \leq 6K\log(n)$  and Lemma~\ref{lem:bd_min_diff_alpha_gen_k} for  $m\geq 6K \log(n)$, as well as the definition of $f(n)$ and $\Pr ( \overline{\Ev}  ) = O(1/n)$  shown above.
\endproof

\end{document}